\documentclass{article}
\usepackage{graphicx} 
\usepackage{amssymb}
\usepackage{algpseudocode}
\usepackage{amsmath}
\usepackage{amsfonts}
\usepackage{amsthm}
\usepackage{amssymb}
\usepackage{adjustbox}
\usepackage{makecell}
\usepackage{algorithm}
\usepackage{algorithmicx}
\usepackage{xcolor}
\usepackage[numbers]{natbib}
\usepackage{subcaption}
\usepackage{booktabs}
\usepackage{makecell}
\usepackage{colortbl}
\usepackage{soul}
\PassOptionsToPackage{hyphens}{url}\usepackage{hyperref}

\usepackage[font=small,labelfont=bf,tableposition=top]{caption}
\DeclareCaptionLabelFormat{andtable}{#1~#2  \&  \tablename~\thetable}

\newtheorem{claim}{Claim}

\newtheorem{proposition}{Proposition}
\newtheorem{definition}{Definition}
\newtheorem{lemma}{Lemma}

\newtheorem{theorem}{Theorem}

\usepackage[final]{neurips_2024}

\def\given{\;|\;}
\title{What Constitutes a Less Discriminatory Algorithm?}
\author{
  Benjamin Laufer \\
  Cornell University\\
  New York, NY \\
  \texttt{bdl56@cornell.edu} \\
  \And 
   Manish Raghavan \\
   Massachusetts Institute of Technology \\
   Cambridge, MA \\
   \texttt{mragh@mit.edu} \\
   \And
   Solon Barocas \\
   Microsoft Research \\
   New York, NY \\
   \texttt{solon@microsoft.com} \\
  }

\date{August 2024}

\begin{document}

\maketitle

\begin{abstract}    
    Disparate impact doctrine offers an important legal apparatus for targeting discriminatory data-driven algorithmic decisions. A recent body of work has focused on conceptualizing one particular construct from this doctrine: the \textit{less discriminatory alternative}, an alternative policy that reduces disparities while meeting the same business needs of a status quo or baseline policy. However, attempts to operationalize this construct in the algorithmic setting must grapple with some thorny challenges and ambiguities.
    In this paper, we attempt to raise and resolve important questions about less discriminatory algorithms (LDAs).
    How should we formally define LDAs, and how does this interact with different societal goals they might serve?
    And how feasible is it for firms or plaintiffs to computationally search for candidate LDAs?
    We find that formal LDA definitions face fundamental challenges when they attempt to evaluate and compare predictive models in the absence of held-out data.
    As a result, we argue that LDA definitions cannot be purely quantitative, and must rely on standards of ``reasonableness.'' 
    We then identify both mathematical and computational constraints on firms' ability to efficiently conduct a proactive search for LDAs, but we provide evidence that these limits are ``weak'' in a formal sense.
    By defining LDAs formally, we put forward a framework in which both firms and plaintiffs can search for alternative models that comport with societal goals.
\end{abstract}

\section{Introduction}
\label{sec:intro}
Over the past decade, most scholarship on algorithmic discrimination that engages with U.S. law has focused on disparate impact doctrine. Unlike disparate treatment doctrine, which requires that decision-makers discriminate intentionally or take legally protected characteristics into account explicitly, disparate impact doctrine allows plaintiffs to challenge facially neutral policies that nevertheless have an unjustified or avoidable disparate impact.\footnote{Disparate impact claims can be raised under the discrimination laws that regulate employment (Title VII), lending (the Equal Credit Opportunity Act (ECOA)), and housing (the Fair Housing Act (FHA)).} Scholars have worried that algorithms will be largely immune from disparate treatment claims because they are not likely to be designed to discriminate intentionally and because they often do not consider legally protected characteristics \citep{barocas2016big,kim2017data,yang2020equal,hellman2020measuring,bent108algorithmic}. Instead, algorithms are more likely to give rise to unintentional, but avoidable disparate impacts.

Disparate impact cases are initiated by a plaintiff who puts forward evidence that the decision-making process at issue has resulted in a disparate impact along the lines of a legally protected characteristic (e.g., race, gender, age, etc.). For example, somebody who was denied a loan would first demonstrate that there is a gender disparity in the approval rates for applicants. The defendant can then put forward what is known as a ``business necessity defense,'' where they claim that the observed disparities are required to meet some legitimate business goal --- for example, that a lender's credit scoring model predicts default to some reasonable degree of accuracy. Finally, the plaintiff can rebut the firm's justification by pointing to what is commonly known as a less discriminatory alternative: an alternative decision-making process that would serve the firm's goals equally well, but with less disparate impact. If a rejected credit applicant can furnish a credit scoring model of equivalent accuracy that has a smaller gap in selection rates across gender groups, then the firm can be found liable for discrimination. 

The final step in disparate impact doctrine and the notion of a less discriminatory \emph{algorithm} (LDA) have attracted significant scholarly attention in recent years, as research suggests that there may be readily available less discriminatory alternatives to existing algorithms used in many contexts. One idea in particular that is attracting growing interest in both the technical and legal communities is the phenomenon of model multiplicity: it is often possible to develop many different models that all exhibit the same accuracy but make quite different predictions on individual points (i.e., a specific person) and groups of points (i.e., specific groups of people) \cite{pmlr-v119-marx20a, black2022model, black2023less, coston2021characterizing, rudin2024amazing, forde2021model, watson2023multi}. Two models might be equally accurate, but make opposite predictions for a specific person. Likewise, the first model might select members from one group more than members of another group, while the second model might select members from both groups at a similar rate, even though both models are equally accurate. Multiplicity therefore suggests that there is not always an inevitable trade-off between accuracy and fairness. Instead, practitioners will often be able to develop a large set of equally accurate models and then choose among these the one that happens to have less disparate impact across groups. Model multiplicity speaks directly to the question of less discriminatory algorithms because it tells us that there will often be many equally accurate ways to select among individuals, some of which will have more of a disparate impact than others.

Legal scholars and computer scientists have worked together to capitalize on the promise of multiplicity for dealing with algorithmic discrimination, exploring the implications of multiplicity for disparate impact doctrine \cite{black2023less}, and proposing specific techniques for finding the least discriminatory alternative \cite{gillis2024operationalizing}. Civil society organizations have likewise sought to leverage this insight, calling on firms to take affirmative steps to find less discriminatory alternatives when developing their decision-making algorithms and asking regulatory agencies to clarify their expectations of regulated firms \cite{NCRC_2022,consumerfed2024,NFHA2024,consumerreports,finreglab}. And a range of federal and state regulators have now issued statements instructing firms to search for less discriminatory alternatives while developing algorithms \cite{cfpb2024,HUD2024,NYS2024,noauthor_guidance_2025}. 

\paragraph*{Open challenges for operationalizing LDAs} Despite this enthusiasm, there remains some uncertainty around LDAs. First, LDAs seem to advance many possible goals: establishing that discrimination has occurred, preventing further discrimination from occurring, and keeping discrimination from occurring in the first place. In the course of a lawsuit initiated by an individual plaintiff, an LDA serves as evidence of an avoidable disparate impact and thus of illegal discrimination. It demonstrates that the defendant needlessly subjected the plaintiff to an adverse outcome that could have been avoided by adopting an equally effective decision-making policy that would have resulted in a less disparate impact. At the same time, LDAs can also prevent future disparate impact by compelling defendants found liable for discrimination to adopt the identified alternatives from that point forward. Finally, LDAs may keep discrimination from happening in the first place: faced with the risk that a future plaintiff might be able to identify an LDA, firms might be incentivized to pro\-actively search for LDAs to head off possible lawsuits. Prior work on LDAs has not been clear about which of these goals they are expected to serve and whether different goals would lead to different interpretations of what constitutes an LDA.

Second, similarly unsettled is the degree to which practitioners can exploit multiplicity to reduce disparate impact. Prior scholarship has acknowledged that there are likely to be limits to how far practitioners can go in leveraging the phenomenon to minimize disparities \cite{islam2021can, black2023less, rodolfa2021empirical}, but stops short of identifying fundamental limits. While Black et al. point out that ``there are technical limits to the extent to which error can be redistributed to reduce disparity,'' they don't pin down the point at which these restrictions begin to apply. If anything, a fair amount of prior work on algorithmic discrimination takes as a given that there is an inevitable trade-off between accuracy and disparate impact \cite{corbett2017algorithmic,menon2018cost,goel2018non,jaramillo2024understanding}. Greater clarity about just how much can be achieved with multiplicity would help to address these claims and any claim by specific defendants that a disparate impact was necessary to achieve its goals.

Third, scholars seem to be unsure whether it's possible to find the \textit{least} discriminatory algorithm \cite{scherer2019applying,black2023less}, on the belief that doing so would require an unbounded --- and therefore intractable --- search. If this is indeed the case, it's unclear what should be expected of an LDA in practice. The difficulty of finding LDAs remains an open question --- and an urgent one, given that firms are likely to invoke this difficulty as a reason why they should not be expected to search for LDAs or why an LDA presented by a plaintiff would have been too difficult for them to discover.

\paragraph*{What we do in this paper.} Here we aim to resolve this lingering uncertainty around what LDAs are meant to achieve, how much can be achieved by exploiting multiplicity, and how easily this can be achieved. We start with a set of conceptual and theoretical results, and then turn to an empirical demonstration on real datasets. Our conceptual contributions are organized as follows:
\begin{itemize}
    \item  When LDAs are litigated, plaintiffs and courts might have more information than what is available when the models are built. In Section \ref{sec:generalizability}, we discuss a variety of statistical challenges arising from this information asymmetry.
    We put forward an LDA definition based on projections of how a model will generalize, as opposed to its measured performance on a realized dataset.
    \item  Mathematically, a classifier can only exhibit certain combinations of accuracy and selection rate disparity between groups, given the size of each group and the base rate of the property or outcome of interest in each group. In Section \ref{sec:polygon}, we provide a strict characterization of what accuracy and disparity levels a classifier can achieve for a fixed population, finding that there exist lower-disparity classifiers in most reasonable settings.
    \item 
    From a computational perspective, the existence of an LDA does not mean one can efficiently find it.
    In Section \ref{sec:computation}, we show formally that a search for a lower-disparity classifier at some baseline level of utility is computationally intractable (NP-hard) in general.
    However, we argue formally that computational hardness is unlikely to be prohibitive in practice.
\end{itemize}
Although our technical results sometimes seem to give firms a natural defense against allegations of discrimination, the negative results only tell part of the story. 
Our theoretical and empirical analysis suggests that there exist effective and efficient strategies for identifying lower disparity alternatives over the set of possible policies, even when the math suggests that the \textit{perfect} (least discriminatory and accurate) policy may be out of reach. 
In Section \ref{sec:empirics}, we provide results from empirical tests of a particular group of methods that randomly generate alternative models in the same model class. In some observed instances, we find that these methods can reduce disparity out-of-sample, sometimes at no cost to utility.

\subsection{Related Literature}

Here we provide an overview of related literature on less discriminatory alternatives, accuracy and fairness, and multiplicity.

\paragraph*{Less discriminatory alternatives and statistical approaches} The applicability of discrimination law to data-driven algorithmic decisions has received much recent attention \cite{barocas2016big,kim2017data,gillis2019big}. Attempts to define and operationalize normatively or legally useful notions of fairness abound \cite{hardt2016equality,dwork2012fairness,kim2018fairness}. In recent work, \citet{gillis2024operationalizing} put forward a definition formalizing the notion of a less discriminatory alternative. In their formulation, optimization occurs over a fixed dataset and both false positive rate and false negative rate are constrained. The authors use a mixed-integer program to identify LDAs in the case of linear classifiers. 
\citet{auerbach2024testing}, similarly inspired by disparate impact law, develop a statistical method for comparing the fairness-improvability of an alternative classifier to a baseline. 
\citet{black2024d}, however, point out that fairness improvements may not generalize --- and that requiring the reporting of fairness improvements on training data could lead to manipulations and faulty results, which the authors term `D-hacking.' 

\paragraph*{Accuracy-fairness trade-offs and welfare} Fair machine learning is a vast area of research and finding satisfactory notions compatible with commonly held normative intuitions has proved curiously difficult \cite{barocas2023fairness,chouldechova2018frontiers,cooper2021emergent,laufer2022four}. Scholars have developed impossibility results, trade-offs, and welfare notions to frame the terms and achievable goals of defining algorithmic fairness
\cite{kleinberg2016inherent,chouldechova2017fair}.
\citet{liang2022algorithmic} put forward a model for understanding the achievable performance across two groups' error rate and overall accuracy. \citet{pinzon2024incompatibility} analyze the feasible set of realized equal opportunity and accuracy and provide conditions where these are incompatible. \citet{pleiss2017fairness}, focused on trade-offs between fairness notions, find that satisfying welfare defined over different error rates is equivalent to treating some population members randomly. Some empirical work has leveraged findings in this space to develop fairness or fairness-accuracy improvements in real-world systems \cite{lee2024ending,liu2024redesigning,buolamwini2018gender}.

\paragraph*{Multiplicity} There is a growing interest in the idea that many models can exhibit similar performance, but differ in the other properties that they possess (e.g., fairness, interpretability, etc.). This notion is referred to commonly as model multiplicity \cite{black2022model,watson2023multi,watson2023predictive}, the Rashomon effect \cite{breiman2001statistical,watson2024predictive}, or under-specification \cite{d2022underspecification}. \citet{semenova2022existence} find that the size of the set of good models can serve as a measure of model-class simplicity. \citet{cooper2023my} find that variance in the design of classification algorithms can lead to individual people receiving different treatments depending on which model is chosen. In the context of disparate impact cases, \citet{black2023less} argue that the onus should be placed on the firm --- rather than the plaintiff --- to conduct a reasonable search for a less discriminatory alternative in light of the flexibility afforded by multiplicity \cite{coston2021characterizing,rudin2024amazing}.
\section{Statistical Nuances in Defining the LDA}
\label{sec:generalizability}

A less discriminatory alternative must satisfy two properties: it must be at least as effective as the original practice from a business perspective, and it must, of course, be ``less discriminatory,'' typically as measured by disparities in selection rates across demographic groups~\citep{black2023less}. When the practices and alternatives in question are decisions based on predictive models, the literature typically suggests defining business efficacy in terms of the model's predictive accuracy~\citep[e.g.,~][]{gillis2024operationalizing,relmancolfax2021}.
These two metrics---predictive accuracy and selection rate disparities---are the primary tools used in the literature to reason about LDAs.

Both of these metrics are \textit{distributional}.
A model's accuracy and disparate impact cannot be determined in the abstract; it can only be determined with respect to a specific dataset. In other words, the accuracy and disparate impact exhibited by a model depends on the data that it is applied to. In order to evaluate a candidate LDA, we must first ask: what data should be used to assess a model's accuracy and selection rate disparities?
Typically, when developing a machine learning model, a firm must build a dataset that (in their judgment) best approximates the distribution of applicants they expect their model will encounter.

Suppose a court wants to evaluate whether a candidate LDA is at least as accurate and less discriminatory than an existing model.
Now that the model has been deployed, courts and plaintiffs have a new source of information at their disposal, one that the firm could not have used during development: They have observed the actual applicants on whom the model was used to make decisions.
Our focus in this section will be on how the information available in these two contexts---during development and after deployment---impacts how we might think about operationalizing LDAs.\footnote{Throughout this paper, we limit the scope of our inquiry to the task of finding a less discriminatory algorithm when there is a predefined training dataset and prediction target. Of course, in the course of an actual lawsuit, a plaintiff could also challenge the dataset and target used by the defendant, arguing that both or either are flawed. We intentionally set aside this possibility for the sake of analytic clarity.}

\paragraph*{Measuring accuracy and disparity after deployment.}
We proceed by analyzing how a court might measure accuracy and disparity in turn.
A key challenge in measuring a predictive model's accuracy on a data distribution is the need for ``ground truth'' labels.
Without access to labels, one can observe the predictions made by a model, but has no ability to assess whether the predictions are correct. 
In the contexts covered by discrimination law, labels (i.e., outcomes) are not reliably observed after a model is deployed~\citep{raghavan2024limitations}. 
For example, in lending tasks where the outcome in question is loan repayment, labels can take years to detect.
Even worse, these contexts suffer from a \textit{selective labels} problem: Outcomes (like repayment) may only be observed for applicants who were selected (granted a loan) \citep{lakkaraju2017selective}.
Conventional measures of accuracy using post-deployment data are not reliable or applicable.\footnote{An active line of research seeks to combine information from pre- and post-deployment data to overcome this challenge~\citep{coston2020counterfactual,coston2021characterizing,rambachan2022robust}.}

Selection rate disparity, on the other hand, can be measured on post-deployment data because it does not depend on ground truth labels. To measure selection rate disparity, one can compute the fractions of each demographic group who were selected based on the model's predictions. It suffices to observe the transcript of the model's predictions for all applicants alongside their demographic information.

\subsection{From measured performance to projected performance}

Given the limitations of post-deployment data discussed above, it might seem reasonable to estimate a model's real-world properties as follows: 
Use the firm's pre-deployment data to measure accuracy, but use post-deployment data to measure disparity.
While this might indeed provide a reasonable picture of realized model performance, we will argue that it is not a workable standard for an LDA.
Simply measuring accuracy or disparity on available data is not sufficient as a basis for defining an LDA. 

Key to this claim is a simple technical insight: One can almost always find a model that outperforms a given benchmark measured on existing data.
In other words, a plaintiff will in general be able to produce a model that achieves values of both accuracy and disparity that are at least as good as the defendant's original practice.
But such a model is clearly undesirable to both firms and plaintiffs, and as a result, one cannot reliably determine whether a model is an LDA simply by measuring its performance.
\begin{proposition}
\label{prop:no-lda}
Under the assumption that no two applicants are identical, we can find a decision rule that simultaneously achieves perfect accuracy on one dataset and zero disparity on another.
\end{proposition}
This observation is by no means deep.
Consider the following decision rule: Select all positively-labeled applicants in the pre-deployment data, and select all applicants in the post-deployment data. (Or otherwise select the same constant fraction of applicants in each demographic group in the post-deployment data.) Without limits on the complexity or flexibility of a decision rule, we can effectively construct one that performs perfectly according to the given measures.

Why do we bring up this example?
The reader might point out that no reasonable court would consider this pathological decision rule to be an LDA.
Aiming to achieve perfect accuracy on a dataset would lead to severe over-fitting, and selecting all applicants would clearly fail to meet the firm's business objectives.
And yet, this decision rule, strictly speaking, meets the standards we have put forward for how to understand the LDA: sufficiently good accuracy measured on the pre-deployment data, and reduced disparity measured on the post-deployment data.
Moreover, by a similar argument, \textit{any} LDA definition that relies entirely on performance measures of this form are vulnerable to the same type of pathological example.

We might hope to minimally amend our LDA definition to exclude decision rules that cherry-pick individuals like the one described above.
Unfortunately, there is no bright line separating the pathological case above and the general class of models learned from data.
As a simple example, a standard nearest-neighbors model will achieve perfect performance on its training set.
The same is true for a sufficiently expressive neural network.

This seemingly absurd example thus reveals a fundamental point: 
Performance metrics measured on a dataset, or even on a combination of datasets, cannot be the sole determinant of what constitutes an LDA.
One can always find a model that optimizes metrics on observed data but clearly should not be considered a viable alternative to the firm's practices.

\paragraph*{Projected performance.}

It is instructive to more closely examine \textit{why} a predictive model that achieves perfect performance on seemingly desirable metrics is nonetheless undesirable. 
When building machine learning models, the goal is \textit{generalization}.
Algorithm designers want the model to perform well not only on existing data, but also on data that they have yet to encounter.
In fact, machine learning procedures are designed to forgo pre-deployment (``in-sample'') performance in an attempt to improve outcomes during deployment (``out-of-sample''), typically by restricting attention to a class of ``simple'' models.\footnote{Recent literature on ``double descent'' has begun to change our understanding of classical bias-variance trade-offs, suggesting that sufficiently overparameterized models generalize well~\citep{belkin2019reconciling}. Still, without held-out data, it is difficult to evaluate whether models generalize.}
To make comparisons between candidate models, practitioners typically maintain a separate ``hold-out'' or ``validation'' dataset, which they do not use directly to train models.
Measuring performance this way prevents them from learning models that are overfit to the training data. Overfit models make accurate predictions on existing data but do not generalize to unseen data.

Evaluating on unseen data is the primary means by which researchers and practitioners can forecast the performance of models once deployed.
Without held-out data, one cannot reliably evaluate or compare models.
But critically, when evaluating a candidate LDA in a court setting, 
\textit{there is no held-out data}.

So how can a court determine whether a proposed LDA would be expected to generalize?

We argue that the way to evaluate a potential LDA is not to \textit{measure} its performance, but to \textit{project} its future performance.
Practitioners project performance by evaluating models on held-out data.
But in the absence of held-out data, projections of performance must rely on heuristics.
We argue that this requires an LDA to be held to a ``reasonableness standard.'' 
A key aspect of this standard is the complexity of the candidate LDA relative to the original model.
Classical results from statistical learning theory tell us that simple (or regularized) models tend to generalize better than complex models, despite the fact that their in-sample performance may be lower~\citep[see, e.g.,~][]{vapnik1998statistical}.
Consider a firm that chooses to build and deploy a very simple model---such as a logistic regression---because they lack sufficient data to reliably train a more complex model.
Suppose a plaintiff presents a court with a highly complex neural network as a candidate LDA.
Indeed, the plaintiff's neural network might outperform the firm's model on all relevant performance- and disparity-related metrics on both pre-deployment and post-deployment data.
Under a reasonableness standard, a court would still need to determine whether the plaintiff's model would be expected to generalize beyond the particular datasets.
Without such a standard, we argue, a plaintiff will virtually always be able to find a candidate LDA that appears to outperform the original model. 

With this, we propose the following LDA definition.
\begin{definition}
Given an original model $h^0$, a candidate model $h'$ is an LDA if $h$ is reasonably projected to have (1) accuracy at least as high as $h^0$, and (2) selection rate disparity significantly lower than $h^0$.
\label{def:lda}
\end{definition}

\subsection{How should we project performance?}

To fully operationalize Definition~\ref{def:lda}, we would need to formally define how one might make reasonable projections of model characteristics.
Doing so in full generality is beyond the scope this work. 
In what follows, we offer some general guidance. We begin with a discussion of projections and their use for comparing models. We then turn to the question of what data to use for projecting performance. Finally, we discuss how these procedures interact with the broad societal goals behind the LDA.

\paragraph*{Comparing models.}

Machine-learned models are typically selected from a pre-defined practitioner-chosen class of models.
This might be, for example, the class of decision trees of depth at most 5.
Or the class of neural networks following a particular architecture.
Subject to this choice of model class, ML training algorithms seek to select a model within the class that fits the training data well.
But this choice is typically not unique; it may depend, for example, on a randomly chosen seed or some other arbitrary choice~\citep{black2022model}.
As is common in the model multiplicity literature, we will illustrate the role of this randomness in Section~\ref{sec:empirics}.

For models that lie in the same model class, we can typically compare \textit{projected} performance by comparing \textit{measured} performance on a representative dastaset, simply because we have no further information with which to distinguish them.
Training algorithms do just this: Out of an infinitely large class of models, they seek to select the one that best fits the data because it also has the best projected performance.
Thus, we can reasonably compare models in the same model class just by measuring their performance characteristics on data.

When models come from different model classes, however, comparisons are less clear.
A simple logistic regression may have worse measured characteristics than a very complex neural network, but due to its complexity, the neural network is more likely to overfit to the data.
Whether or not we believe the regression should have better \textit{projected} characteristics depends a wide range of factors like the amount of data on which the models were trained and the complexity of that data.
In general, comparisons of this form cannot be conclusive and must be handled on a case-by-case basis.

\paragraph*{On what data should projections be based?}
Should projections be based on information obtained pre-deployment or post-deploy\-ment?
This choice reveals an additional subtlety.
During model development, firms (by definition) only have access to pre-deployment data. 
Plaintiffs, on the other hand, can in principle observe both pre- and post-deployment data.\footnote{This, of course, assumes that plaintiffs can successfully access the data through a discovery process; if not, plaintiffs may struggle to develop candidate LDAs.}
If projections are based solely on pre-deployment data, this asymmetry has little impact.
But if projections take post-deployment data into account, plaintiffs could potentially use this information to guide their LDA search.

Suppose, for example, a firm has identified two potential models, $h_1$ and $h_2$, during development and seeks to select between them.
Using their pre-deployment data, they find that the two models have (roughly) identical projected performance.
Absent further information, they cannot determine which will be less discriminatory; suppose they arbitrarily choose to deploy $h_1$.
After deployment, however, suppose projections based on a combination of pre- and post-deployment data suggest that $h_2$ is less discriminatory than $h_1$.
If they base their LDA projections on this combined data, a plaintiff will be able to put $h_2$ forward as an LDA, despite the fact that prior to deployment, the firm had no way to distinguish between them.
Thus, projecting model characteristics using post-deployment information may allow plaintiffs to engage in a form of ``D-hacking''~\citep{black2024d}, where they select between models that look identical or near-identical at development time.
Indeed, the model multiplicity literature suggests that the set of models with identical projected characteristics based on pre-deployment data will take on a range of projected characteristics when post-deployment data are factored in~\citep{black2022model}.

\paragraph*{Serving societal goals.}
Armed with these general insights, we return to the question raised in Section~\ref{sec:intro}: Is the court's goal (1) to establish evidence of past discrimination, (2) to remedy decision-making for future outcomes, or (3) to discourage discrimination from occurring in the first place? Focusing, for a moment, only on the first of these goals, it might seem to make sense to use post-deployment data---after all, the post-deployment data is about the people affected by the alleged disparate impact. We have argued, however, that post-deployment data is not a reasonable basis for establishing whether firms did something wrong. 
Instead, the first of these goals is best served by projecting characteristics based on pre-deployment data. The use of pre-deployment data alone best distinguishes between desirable and undesirable model development practices. Projecting using pre-deployment data also seems well-equipped for serving goal (3). If a firm can reliably test for LDAs prior to deployment, they can proactively search for LDAs to both minimize the risk of future liability and promote social non-discrimination goals.\footnote{\citet{black2023less} argue for precisely this sort of proactive ``reasonable search'' on the part of firms.}

On the other hand, if LDAs are intended to change ongoing practices to prevent future discrimination (goal 2), the most accurate projections should take into account all available data, including post-deployment data.
In the hypothetical example above, once we learn that $h_2$ indeed leads to lower post-deployment disparity than $h_1$, should we not expect it to have lower disparity on unseen data?\footnote{Of course, it could be the case that $h_2$'s reduced disparity happens to be a statistical artifact of the post-deployment data that does not generalize. Still, absent further information, $h_2$ is more likely than not to have lower disparity on future data than $h_1$~\citep{recht2019imagenet,salaudeen2024imagenot}.}
Even so, incorporating post-deployment data may make some uneasy.
After all, in our example, the firm had no information to suggest that $h_2$ would be preferable to $h_1$.
Would the use of post-deployment data lead the firm to be held liable for consequences they could not have reasonably foreseen?
We do not attempt to resolve this tension between discouraging bad practices and improving them for the future.
But by adopting a particular LDA definition, a court will have to navigate this tension, implicitly or explicitly.

\paragraph*{Other uses of post-deployment data.}

As described in Section~\ref{sec:intro}, the scope of this definition does not preclude challenges to the firm's practices on other grounds.
A plaintiff might reject the firm's pre-deployment data altogether as unrepresentative of the true applicant distribution.
Post-deployment data may be useful in supporting such claims.
They might also raise issues with the overall predictive setting, including the features collected and the validity of the construct being predicted~\citep{raghavan2020mitigating}, or they might provide a less discriminatory alternative wholly unrelated to a predictive model.

\section{Mathematical Limits}
\label{sec:polygon}

Defining LDAs is only half the challenge.
Firms seeking to proactively defend against LDA claims will need to search for LDAs meeting that definition.
In what follows, we explore fundamental limits on our ability to efficiently search for LDAs, corresponding to objections a firm might raise as to why it was unable to find an LDA.
Mathematically, certain accuracy-disparity combinations may be unachievable, leading to a necessary trade-off between these two objectives.
And even if an LDA exists, it might be computationally intractable for a firm to find it.
We discuss these potential objections in Section~\ref{subsec:polygon} and Section~\ref{sec:computation} respectively.
We find that while they each have some merit, they are both weak, in a sense we can make formal.
The mathematical limits on accuracy-disparity combinations are unlikely to be binding in practice, as they only hold for algorithms that are already highly accurate.
On the computational side, while it is computationally hard to find the \textit{least} discriminatory algorithm, this does not extend to finding an algorithm that is \textit{nearly} the least discriminatory.

The broad strategy of our analysis in this section is to make assumptions that are generous to the firm.
Hardness or impossibility results will thus hold even under less generous assumptions.
In particular, we will assume that the firm has full information about the population it selects from, as opposed to the standard machine learning assumption that the firm has \textit{samples} from this population.

\subsection{A Working Formalism}
\label{subsec:formalism}

Here we put forward a mathematical formalism for our setting, in which a firm makes a selection among a population. Our aim here is to introduce notation that will support inferences about what is attainable -- and what is not attainable -- in the search for a less discriminatory algorithm.

\paragraph*{Population and classifiers.} We have a finite population of people, where each member (e.g., a loan applicant) is described by categorical feature value $x$, binary group $g\in\{1,2\}$ and binary true label $y \in \{0,1\}$. Interchangeably, we may refer to the class labels as negative `$-$' and positive `$+$'. The full set of data values, labels, and groups are given by vectors $X, Y,$ and $G$, respectively. The size of the population is $n:=|X|$ and we may refer to subsets of the population using subscripts $n_{g,y}$. 
We define a classifier $h(x): x \rightarrow \{0,1\}$ as a mapping from features to binary labels. 
The particular classifier that a firm commits to is $h^0(x)$, and a candidate alternative classifier (which may or may not be an LDA) is $h'(x)$.

\paragraph*{Selection rates and errors.} For a given classifier $h$, denote the selection rate $\texttt{SR}(h):=\textbf{P}_x[h(x)=1]$. For a particular group $g$, the group-specific selection rate would be $\texttt{SR}_g(h):=\textbf{E}_x[h(x) \given g]$. Finally, because the firm is hoping to design their classifier $h(x)$ to mimic the value $y$, we define the standard notions of \textit{false positive rate}, \textit{true positive rate}, \textit{false negative rate}, and \textit{true negative rate}. We will refer to them using their 3-letter abbreviations. 
The group-specific false positive rate $\texttt{FPR}_g$ is accordingly defined as the $\texttt{FPR}$ conditional on group membership $g$.

\paragraph*{Disparity and utility.} We define the demographic disparity of a classifier $h$ to be $\mathbf{\Delta}(h) := \texttt{SR}_1(h)-\texttt{SR}_2(h)$. 
Sometimes, we'll be interested in absolute disparity, which we define as the absolute value of the demographic disparity. 
Finally, here we offer a notion of \textit{utility} for a given classifier $h$, which represents a firm's interest in accurate performance over the set of classifiers. 
We define utility as $\mathbf{U}(h;\lambda):=\texttt{TPR}-\lambda \texttt{FPR}$ where the value of $\lambda \in \mathbb{R}^+$ suggests the relative benefit of a true positive compared to the cost of increasing the probability of false positive classification.

\subsection{Feasible Set}
\label{subsec:polygon}

In a world where firms risk facing liability because they failed to find an LDA, we might imagine a firm raising certain objections to the notion of LDA we put forward in the prior section. One objection a firm might make is that there is an inevitable trade-off between reducing disparity and performance (accuracy, utility, etc.). At high accuracy, it is impossible to achieve 0 demographic disparity, so these goals can be in tension. Here we show that, indeed, there is some limit on a classifier's accuracy beyond which it is impossible to achieve a 0-disparity classifier, as long as the populations have different base rates. 

This section's approach is to consider the set of all attainable accuracy and disparity values for a binary selection rule. 
The set of attainable accuracy and fairness values can be surmised given access to a dataset $X$ with group membership $G$ and outcomes $Y$. If we can peek at the outcomes corresponding to every data point in a dataset, then we can produce decisions that bound the attainable accuracy and/or disparity performance. For example, the perfectly accurate decision rule could be attained by simply using the outcome data $Y$ as classification labels. 
The perfectly biased decision rule would only select based on protected attribute $G$. The highest-accuracy, lowest-disparity decision rule could be identified by starting with the perfect-accuracy classifier and swapping either 1) positive-labeled advantaged group members or 2) negatively-labeled disadvantaged group members until selection rates equalize (choose whichever swap optimally trades off utility for disparity reductions). 

\begin{figure}[t!]
        \centering
        \begin{subfigure}{.19\linewidth}
            \centering
            \caption{Example $G,Y$}
            \vspace{1.235cm}
            \begin{tabular}[b]{c|c|c|}
          $\ $  & \textbf{+} & \textbf{-} \\\hline
          \rule{0pt}{2ex}  
 		  \textbf{1} & 15 & 20 \\\hline 
     \rule{0pt}{2ex}  
          \textbf{2} & 5 & 10 \\\hline
    \end{tabular}  
    \\
    \begin{minipage}{.5cm}
            \ \\
            \ \\            \ \\
            \ \\            \ \\
            \ \\
            \ \\                        
            \ \\            
    \end{minipage}
        \end{subfigure}
        \hfill
        \begin{subfigure}{.405\linewidth}
            \centering
            \caption{Feasible $\mathbf{\Delta},\mathbf{U}$ (random)$\ \ \ \ \ \ \ \ \ \ $}
            \vspace{0.06cm}
             \includegraphics[width=1\linewidth]{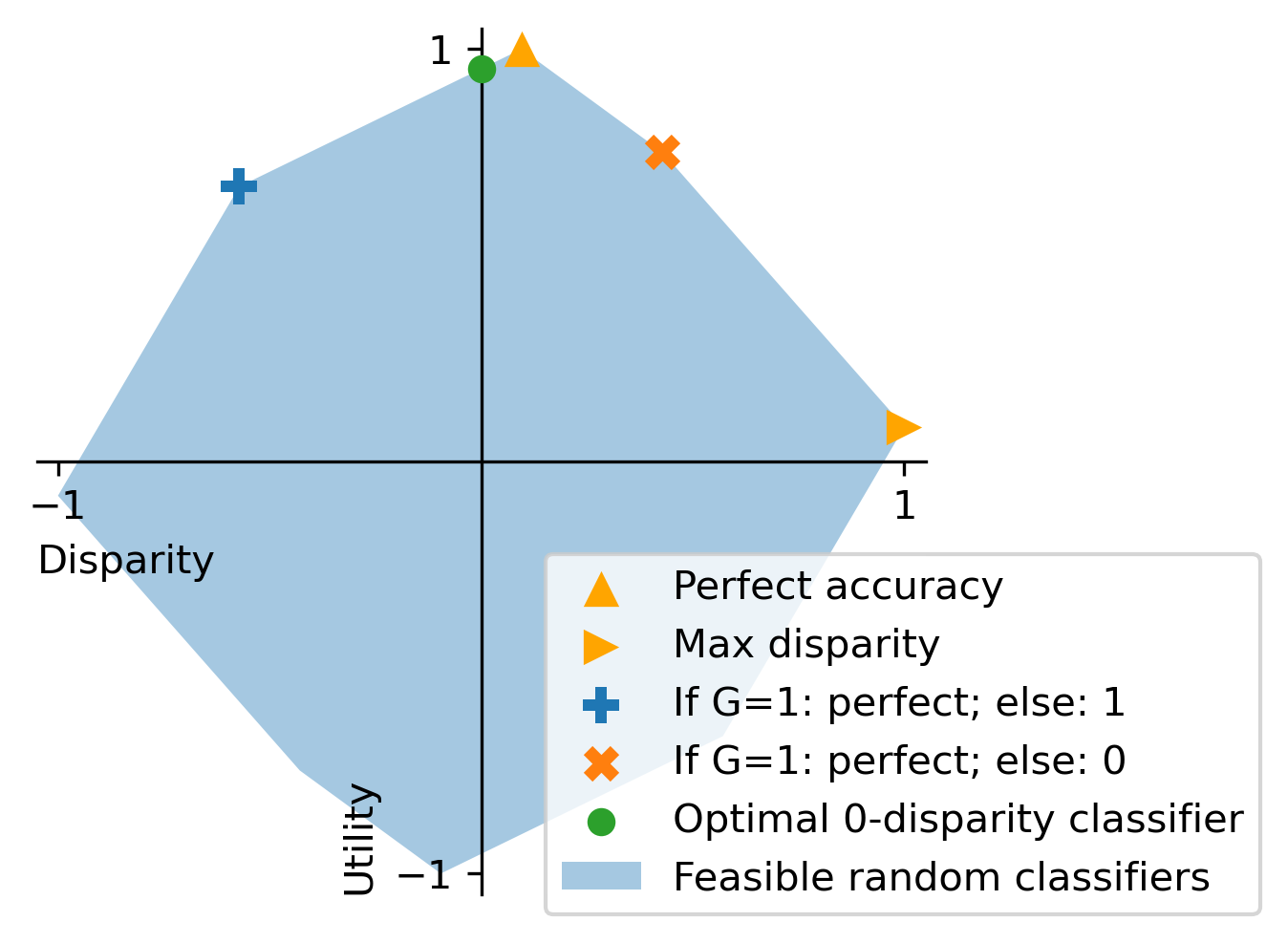}
        \end{subfigure} 
        \begin{subfigure}{.3175\linewidth}
            \centering
            \caption{Feasible $\mathbf{\Delta},\mathbf{U}$ (deterministic)}
             \includegraphics[width=1\linewidth]{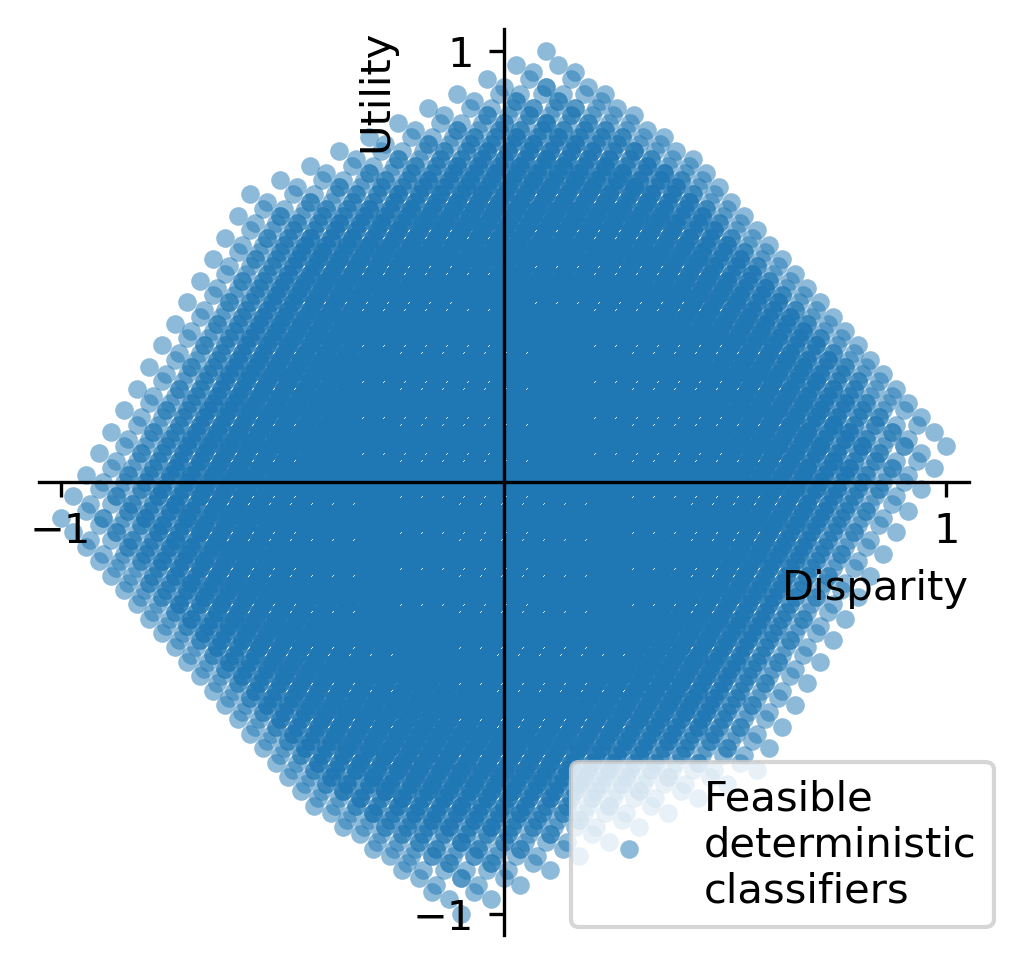}
        \end{subfigure} 
        \vspace{-1.2cm}
        \caption{Consider a given, finite population broken down by group belonging and outcomes (a; left). If randomized decision rules are feasible, then a polygon depicts the convex set of feasible, in-sample utility and disparity values (b; center). If solutions are restricted to deterministic classifiers over the dataset (c; right), the polygon bounds the achievable values.}
        \label{fig:polygon-example}
\vspace{-.5cm}
\end{figure}

A representation of the feasible set of all decisions on the utility-disparity plane is provided in Figure \ref{fig:polygon-example}. The blue shape in Figure \ref{fig:polygon-example}(b) is the region encompassing all achievable accuracy (i.e., utility) and disparity values. Note, however, that as long as the data and decision are discrete, not all values within the polygons are achievable. With discrete data, the set of achievable decision rules is not connected. Instead, it is a dense grid of points. Moving from one point to a neighboring point corresponds to switching the label on a single individual in the population. This grid is depicted in Figure \ref{fig:polygon-example}(c). For exposition, in much of the remainder our analysis, we allow randomized decisions, though the results should hold for the discrete case.

Any reasonable classifier should be in the upper right or upper left quadrant of the disparity-accuracy plane. 
Looking at Figure \ref{fig:polygon-example}, if two of our goals are to maximize utility and minimize disparity, then an ideal classifier should be as far north as possible, without veering too far west or east. Notice that the most accurate classifier requires some non-zero level of disparity: if a classifier exhibits perfect accuracy on a dataset, then its disparity is the difference in group-specific base-rates (this point is plotted as a yellow pyramid). 
Among the set of 0-disparity classifiers, the most accurate is plotted as a green point. 
There are a number of questions we can ask about this feasible set. For example, how high does accuracy have to be such that there does not exist an alternative classifier with 0-disparity and the same (or higher) firm utility? How high does accuracy have to be before we are unable to achieve a disparity in selection rates no greater than 10\%? 
In the remainder of this section, we make formal claims about the properties of the accuracy-disparity polygon.

\begin{theorem}
\label{claim:boundonfairness}
Assume group 1 has a higher base rate and randomized classifiers are feasible. 
There is a utility threshold 
\begin{equation*}\mathbf{U}^* = 1-\min \left[\frac{n_1}{n_+},\frac{\lambda n_2}{n_-}\right]\left(\texttt{BR}_1-\texttt{BR}_2\right)
\end{equation*}
such that there exists a zero-disparity alternative classifier $h'$ with $\mathbf{U}(h') \ge \mathbf{U}(h^0)$ if and only if $\ \mathbf{U}(h^0) \le \mathbf{U}^*$. Additionally, if $\ \mathbf{U}(h^0) > \mathbf{U}^*$, then the minimum disparity alternative classifier has $\mathbf{\Delta}(h') = 1-\min \left[\frac{n_1}{n_+},\frac{\lambda n_2}{n_-}\right]\left(\mathbf{\Delta}(h^*)-\mathbf{\Delta}(h^0)\right)$.
\end{theorem}

The proof of the above finding is provided in Appendix \ref{app:polygon}. Intuition for this result can be found in Figure \ref{fig:polygon-example}(b). The classifiers that optimally trade off between utility and disparity can be found along a (small) line segment between the perfect-accuracy classifier and the optimal, 0-disparity classifier. As long as a starting classifier has utility less than the optimal, zero-disparity classifier, then, there exists a zero-disparity LDA with equal or greater utility.

When base rates differ, the accuracy-optimal classifier will have non-zero disparity. It is plausible to imagine that this fact could come up in a disparate impact case. A (hypothetical) firm defending against charges of discrimination might point to this as evidence that there is an inevitable trade-off between accuracy and minimizing disparity, and that therefore, searching for an LDA runs counter to business needs. 
Our findings here refute this claim. The fundamental trade-off between these values on a fixed dataset 
does not `kick in' unless the accuracy of the starting classifier is above a certain threshold, and even still, there is often room to maneuver to minimize disparity. Further, as depicted in Figure \ref{fig:polygon-example}, and as we will observe on empirical datasets in Figure \ref{fig:empirics-polygon}, this fundamental bound is vacuous except at extremely and often unreasonably high levels of accuracy. It is unlikely, therefore, that this bound would offer a sound defense to firms in realistic settings. 

Although Figure \ref{fig:polygon-example} depicts all feasible values of utility and disparity, not all points depicted can be achieved by a reasonable model. Our findings on the existence of feasible, 0-disparity alternative classifiers should not be read as suggesting that these models are always available or easy to find. Rather, these results establish that a particular argument --- that there exist no alternatives with lower disparity --- does not hold weight in many cases.  
\section{Computational Hardness}
\label{sec:computation}

In this section, we address another objection that firms may raise about our LDA definition: that finding the least discriminatory algorithm is resource-intensive and burdensome, and as a result, they should not be held to that standard. Here we analyze the computational burden, and find that this claim has some merit. 
It can indeed be computationally hard to find the least discriminatory algorithm. However, the notion of hardness we establish in this setting is relatively weak. Nothing prevents firms from finding algorithms that are very close to being the least discriminatory. As long as the LDA definition allows for a bit of wiggle room (e.g., an LDA has to be significantly better than the original algorithm), this objection does not hold. 

\paragraph*{Distributional setting.} This section, like the previous one, identifies certain limits in what the firm is able to achieve. In order to establish these limits, we allow ourselves to make \textit{generous} assumptions about the firm's capabilities and the information it is able to access. If a task is hard for a firm endowed with a generous set of tools and information, it is at least as hard for a firm in the real world. Here we show our results for the general setting where the firm accesses full information about the population's distribution. This means the firm is able to access the joint distribution of feature values $x$, group belonging and outcome. The firm can therefore access the following quantities relevant for the LDA search: all possible data values $\mathbf{X}$, the Bayes-optimal predictor over the data values $\sigma(x):= \textbf{P}[y=1\ |\ x]$, and the group-specific probability density $\rho_g(x) := \textbf{P}[x\ |\ g]$.

\paragraph*{The LDA problem.} Because we treat the LDA problem with some formality in this section, it is worth providing a technical definition for the sake of analysis. 

\begin{definition}[Full-information LDA]
\label{def:full-info-LDA}
        Given a set of data values $\mathbf{X}$, a Bayes-optimal predictor $\sigma(x)$, a group-specific probability density function $\rho_g(x)$, a utility function $\mathbf{U}(h;\lambda)$, and a baseline classifier $h^0(x)$, a \textbf{full information LDA} is a classifier $h'$ with utility at least $\mathbf{U}(h^0)$ and absolute disparity less than $|\mathbf{\Delta}(h^0)|$.
\end{definition}

We call this the `full information' setting because we assume that the firm knows the true labels associated with all candidates.
Our general proof procedures and findings would hold in either setting.
We note that our operationalization differs slightly from an alternative put forward by \citet{gillis2024operationalizing} in both our definition as accuracy as well as the class of decision rules we consider.  

\paragraph*{LDA Complexity}. Here we prove that the problem of certifying whether there exists an LDA, as defined in Definition \ref{def:full-info-LDA}, is NP-hard. 
\begin{theorem}
\label{thm:NP-complete}
    The full-information LDA problem $\langle
    \mathbf{X}, \sigma, \rho_g, h^0
    \rangle$ is NP-complete.
\end{theorem}

Appendix \ref{app:NP-hard-proof} contains a proof of the theorem involving a reduction from the Subset Sum Problem.

\paragraph*{Implications}. 
How should we interpret this result? On its face, it seems to provide firms with a natural defense: it is too computationally burdensome to find an LDA in general. A closer look, however, suggests that NP-hardness may not be as much of an obstacle as it might seem. In Appendix~\ref{app:LDA-approximation}, we show there exists a $(1+\epsilon)$-approximation algorithm for the full-information LDA problem. While it may be difficult to find the \textit{least} discriminatory alternative, in the full-information setting, a firm can efficiently find a close approximation.\footnote{We stress that our results do not imply that a firm can efficiently learn an approximate least discriminatory algorithm using a traditional machine learning setting (i.e., empirical risk minimization over some limited class of models); on the other hand, a line of literature shows how to do precisely this for a variety of ML model classes~\citep[e.g.]{kamishima2011fairness,zafar2017fairness,lipton2018does}.}
This is significant because not all NP-hard problems admit arbitrarily close approximations.
In other words, suppose we amend our LDA definition to require that an LDA improves on existing practices by at least some \textit{de minimis} amount $\epsilon$. Even if a firm is unable to find the least discriminatory algorithm, as long as they are within $\epsilon$ of the least discriminatory algorithm, they do not need to worry about a plaintiff furnishing an LDA.\footnote{Beyond the computational reason described here, there are statistical reasons to require LDA improvements to be sufficiently large.} Conversely, if a plaintiff is able to find a significant improvement, we might conclude that the firm did not conduct a thorough enough proactive search for an LDA.

\section{Empirical Results}
\label{sec:empirics}

We conclude with an empirical analysis of a hypothetical LDA search. The aim of our experiments is to better understand the effectiveness and characteristics of certain simple heuristic search methods for finding an alternative classifier with lower disparity and higher utility performance on new data. We do not expect these methods to identify a Pareto-optimal, fair-and-accurate classifier on unseen data --- an algorithm designer would need to be unfathomably lucky to arrive at a classifier of that sort, even on relatively small datasets. Instead, the methods we test search for alternative classifiers that are similar to the starting classifier, by changing the random seed or randomly sampling and re-training. After producing 100 similar models, the methods we test evaluate the level of disparity using held-out data (i.e., an evaluation set) and select the minimum-disparity alternative.

\begin{figure*}[t]
        \centering
        \begin{subfigure}{.302\linewidth}
            \includegraphics[width=\linewidth]{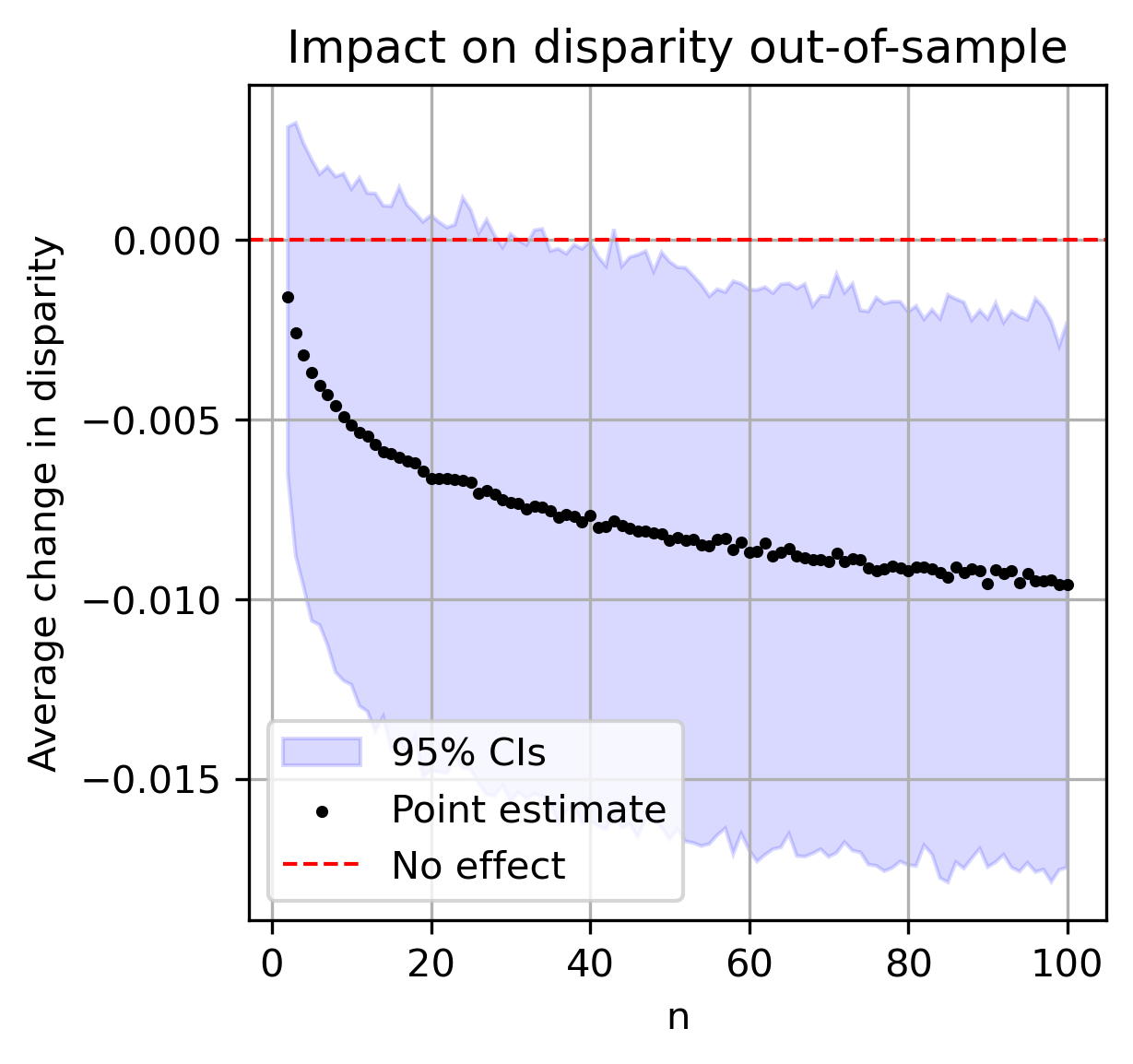}
        \end{subfigure}
        \begin{subfigure}{.294\linewidth}
            \includegraphics[width=\linewidth]{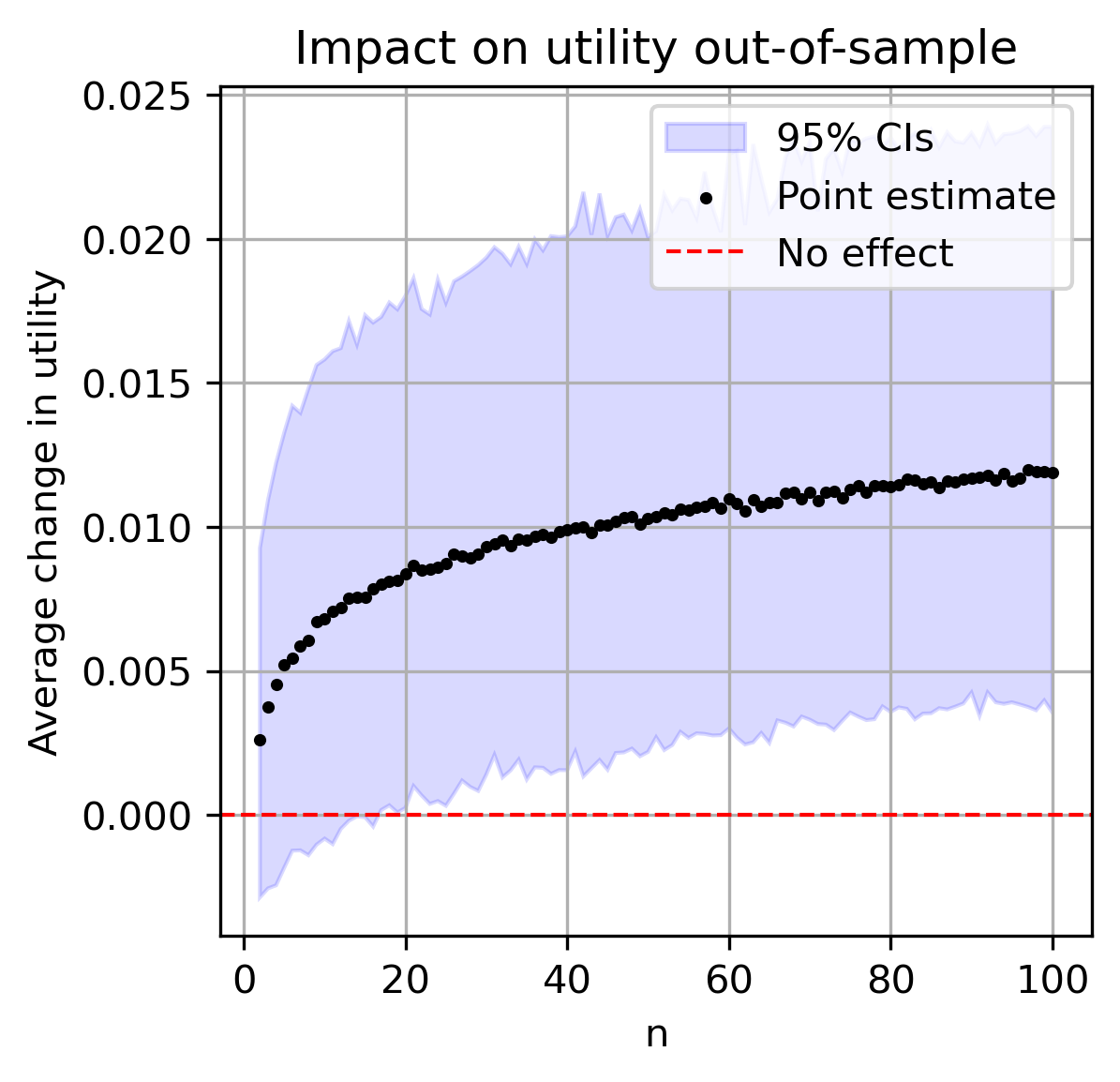}
        \end{subfigure} 
        \begin{subfigure}{.322\linewidth}
            \includegraphics[width=\linewidth]{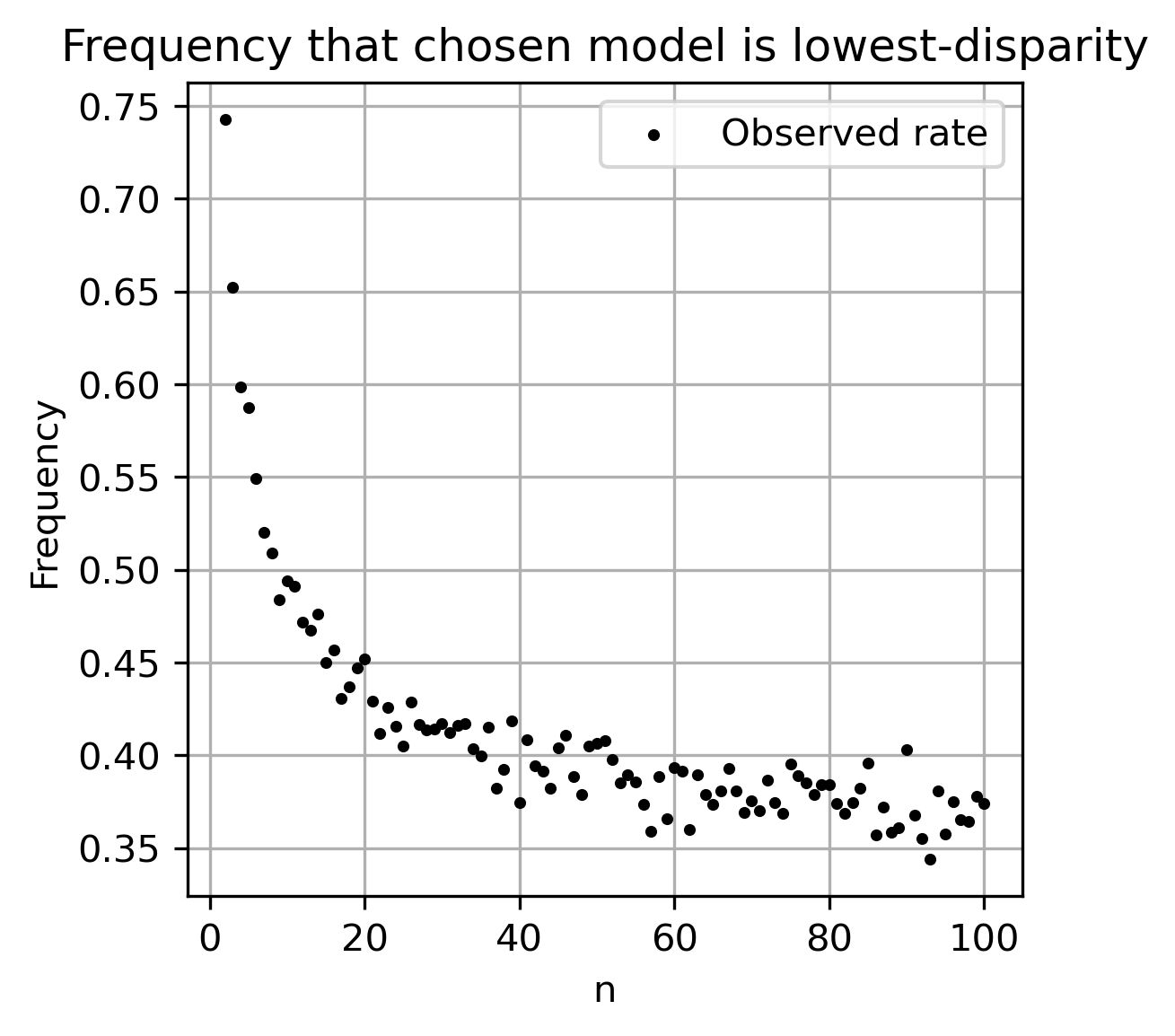}
        \end{subfigure}
        \caption{Results from a simple randomized search for a less discriminatory algorithm on the \texttt{Adult} dataset. The search procedure randomly samples with replacement from the training dataset and retrains a Random Forest classifier $n$ times, for $n \in \{2,...,100\}$. From the $n$ candidate models, it selects the minimum-disparity classifier, as measured using separate data from an evaluation dataset. As $n$ increases, we find disparity decreases, on average, on held out (out-of-sample) test data (a; left). In this case, we also find that utility does not diminish from this procedure (b; center). Confidence intervals are produced by repeating the procedure 2000 times. As the number of random draws increases, the probability of having selected the lowest-disparity model out-of-sample decreases (c; right) suggesting that conducting an effective search might require passing over models that end up having lower disparity.}
        \label{fig:empirics-rf-sample-adult}
\end{figure*}

Although these methods only search within a narrow slice of possible classifiers, they have some desirable properties that make them easy to work with and analyze. First of all, they do not explicitly model disparity and they produce a selection policy that does not, explicitly, use the protected attribute \citep{lipton2018does}. 
Second, the methods produce a set of alternative models that are in the same model class as the original classifier. This suggests that we can reliably compare projected performance (both accuracy and disparity) by measuring them on a fixed dataset. As a consequence, their disparity and utility performance on held-out evaluation data will represent identical and independent draws from a single distribution. This gives the algorithm designer good reason to select the \textit{minimum disparity} classifier on evaluation data as a best-guess for the best-performing model in general.

The code used to implement these tests and produce this paper's figures is available in an online repository.\footnote{Online repository: \href{https://github.com/bendlaufer/what-constitutes-LDA}{https://github.com/bendlaufer/what-constitutes-LDA}.}

\paragraph*{Data and models} For this test, we use two datasets, \texttt{Adult} and \texttt{German Credit}. Both datasets are collected in financial settings. We use the \texttt{Adult} dataset to perform a classification where the task is to predict which individuals make over \$50,000 in income per year. We use the \texttt{German Credit} dataset to perform a classification where the task is to predict credit-worthiness (as defined in the dataset's meta-data). We test three starting classifiers: a Logistic Regression, a Random Forest, and a Decision Tree classifier. We implement these classifiers in Python using \texttt{sklearn}, typically with pre-set default hyper-parameters (though we do set  \texttt{max\_depth}=5 for the random forest classifier). Though the datasets do not have balanced class weights, we train our models using balanced weights.
 A more detailed description of the datasets and models is provided in Appendix \ref{app:empirics}.

\paragraph*{Search processes} We conduct two search types to identify proximate models from the same model class. The first is sampling and re-training. This method involves taking a random sample, with replacement, the size of the training data and re-training the same model on this new sample of data. The second method involves re-training the same model using different \texttt{random\_seed} parameter values. We only use this second method for the Random Forest model because it is an ensemble model that explicitly uses pseudorandomness and the model changes significantly when we perturb the \texttt{random\_seed} parameter.

\begin{table*}[t]
    \centering
    \caption{Out-of-sample performance of various LDA search procedures}
    \begin{tabular}{@{}lllccc@{}}
        \toprule
        \textbf{Dataset} & \textbf{Model} & \textbf{Search} & \textbf{Disparity} & \textbf{Utility} & \textbf{Freq. min-disp.} \\ \midrule
        $\texttt{Adult}$ & Decision Tree     & Sample      & \cellcolor{green!20}-0.0453* & -0.0110 & 0.7310 \\
              & Logistic Reg. & Sample      & \cellcolor{green!20}-0.0170* & \cellcolor{red!20}-0.0189* & 0.6390 \\
              & Random Forest & Random Seed & \cellcolor{green!20}-0.0058* & \cellcolor{green!20}0.0130* & 0.2980 \\
              &     & Sample      & \cellcolor{green!20}-0.0096* & \cellcolor{green!20}0.0119* & 0.3740 \\ \midrule
        \texttt{German} 
                      & Decision Tree &  Sample      & -0.0019 & -0.0008 & 0.0105 \\
                                  & Logistic Reg. & Sample      & -0.0136 & -0.0061 & 0.0175 \\
                                  & Random Forest & Random Seed & 0.0015 & -0.0032 & 0.0040 \\
                                  &     & Sample      & -0.0061 & 0.0034 & 0.0125 \\ 
        \bottomrule
    \end{tabular}
    \label{tab:lda}
\end{table*}

\paragraph*{Experiment procedure} First, we split the data into train, evaluation, and test sets. Using the training data, for each search process and model type, we train 10,000 models, generated at random, and record their performance and qualities in a dataset. We will treat this set of 10,000 models as the space over which a firm is searching for LDAs. The utility, disparity, and selection rate is measured on training, evaluation and testing data for every model, and recorded in a dataset. To test the procedure of searching over $n$ alternative models, we draw from our dataset $n$ times for $n\in \{1,...,100\}$. Equipped with a subset of $n$ models, we select the model with lowest disparity on the evaluation data, and record its utility and disparity performance on test data. The average lift in utility or reduction in disparity is attained by comparing this selected model to the average performance over the $n$ models. For each value of $n$, we perform this sub-sampling procedure 2000 times, and record the mean, 2.5th and 97.5th percentiles to produce point estimates, lower-bounds, and upper-bounds, respectively.

Finally, we also track the rate at which the procedure tested produces a \textit{perfect guess} of the out-of-sample disparity-minimizing model, over the set of models tested. That is, if we conduct our procedure with $n=100$, we measure the frequency that the lowest-disparity model according to evaluation data is also the lowest-disparity model on the test dataset. 

We note that we do not claim to have the best possible procedure for searching, nor do we claim that the method we put forward is \textit{reasonable} in the legal interpretation as advocated by \citet{black2023less}. Instead, we aim to explore whether, in certain instances, simple methods can attain generalizable reductions in discrimination at no cost to accuracy (i.e., utility). We also wish to test whether these methods consistently arrive at the \textit{least} discriminatory classifiers out-of-sample, or whether they open up the possibility that firms test and reject alternatives that end up with lower disparity in hindsight.

\paragraph*{Results} The main results are reported in Table \ref{tab:lda}. The tests provide evidence that certain search methods can reveal the existence of lower-disparity alternative classifiers whose performance generalizes out-of-sample. The asterisk * signifies that the results are statistically significant according to the boostrapped 95\% confidence intervals, attained by repeating the search procedure 2000 times. In other words, the asterisk tests whether the directional change in disparity or utility was observed in at least 95\% of the 2000 times the procedure was tested. Disparity reductions are statistically significant according to this procedure on the \texttt{Adult} dataset for all models and search methods tested. The disparity effects on the \texttt{German Credit} dataset are not significant, which is explained both by the much smaller sample size and the minuscule starting disparity level.

The Random Forest model has the highest utility measure on both datasets compared to the other models tested. It also had the highest absolute disparity on both datasets. These results are depicted in Figure \ref{fig:empirics-polygon}. Performing an LDA search procedure on the RF \textit{increases} utility in our $\texttt{Adult}$ test, meaning in this case, utility increases and disparity decreases statistically--- no trade-off is observed. The observed differences in these values is visualized in Figure \ref{fig:empirics-rf-sample-adult}. In other cases, with other models, however, the utility decreases. This makes sense, given nothing about our procedure aims to maximize the utility. 

\begin{figure*}[t!]
        \centering
        \begin{subfigure}{.468\linewidth}
            \includegraphics[width=\linewidth]{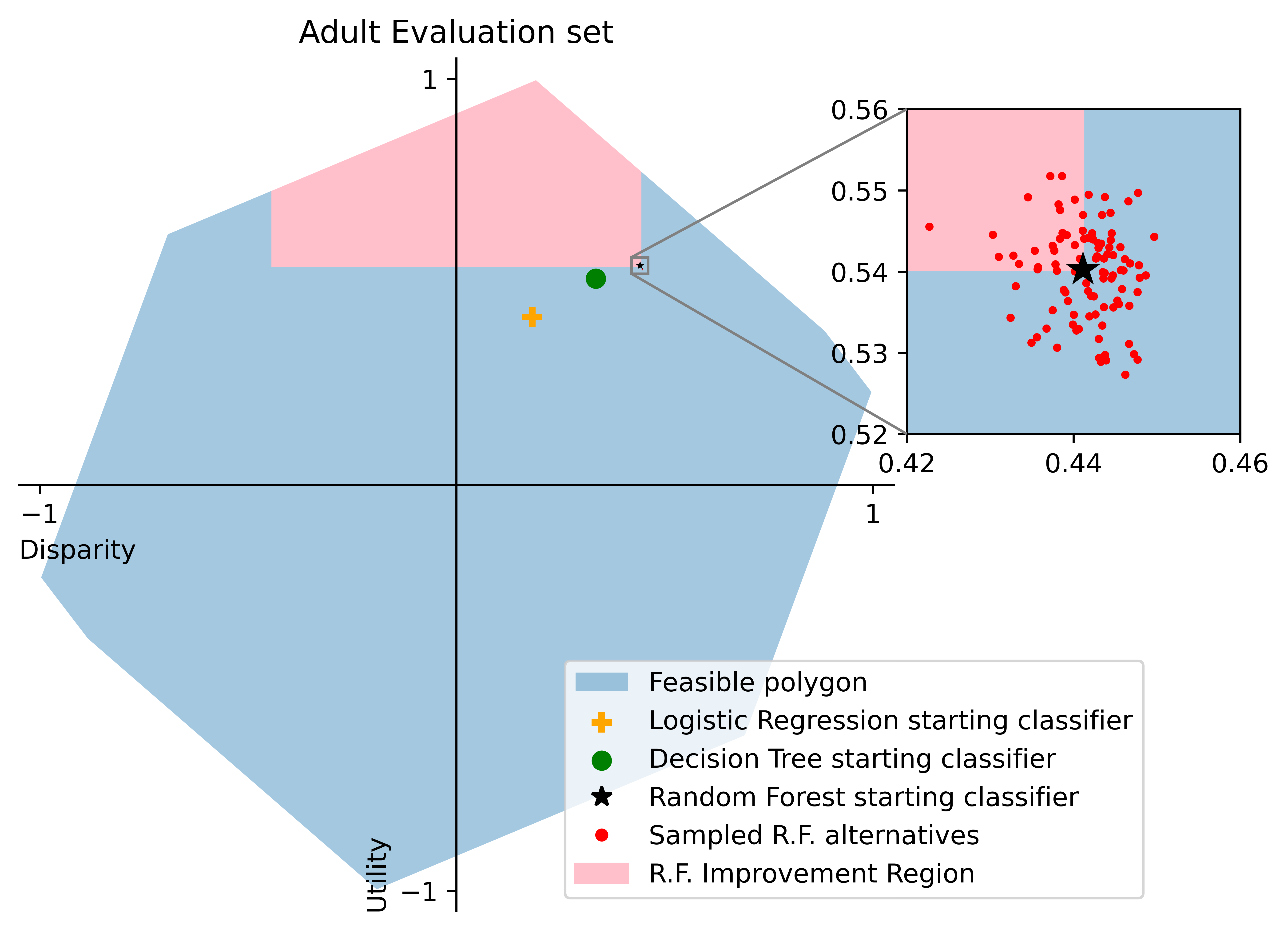}
        \end{subfigure}
        \begin{subfigure}{.414\linewidth}
            \includegraphics[width=\linewidth]{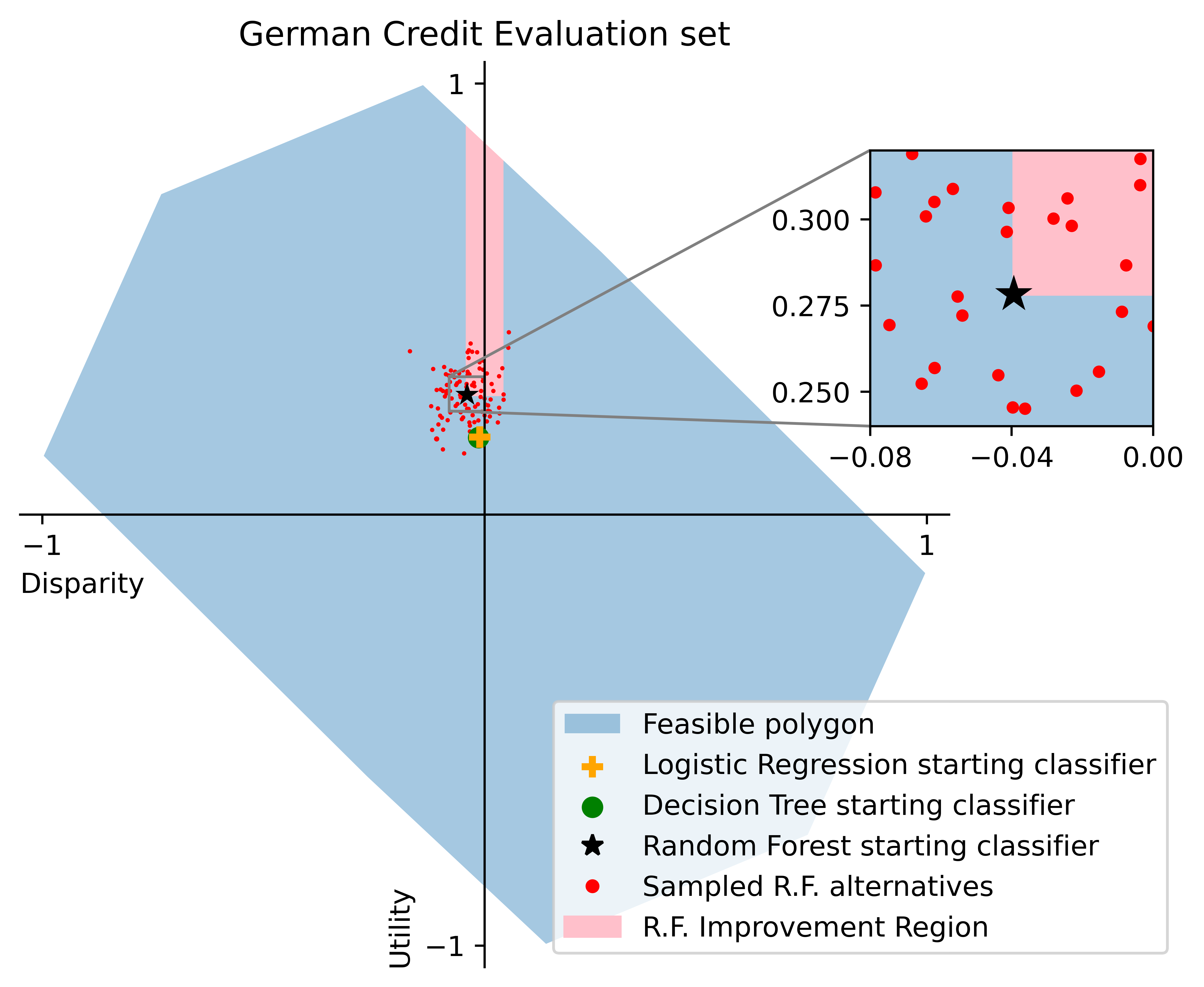}
        \end{subfigure} 
        \caption{Empirically observed achievable polygon using the $\texttt{Adult}$ and $\texttt{German Credit}$ evaluation datasets. In the case of \texttt{Adult} (left), a Random Forest classifier exhibits wide selection rate disparities. Randomly sampled alternatives exhibit small variations in utility and disparity. Analysis of performance on held out test data suggests that selecting the best alternative can statistically improve utility and reduce disparity. In the case of \texttt{German Credit} (right), an initial classifier does not exhibit wide disparities, leaving a narrow region for LDA improvement. No statistically significant disparity reduction is observed on held out test data.}
        \label{fig:empirics-polygon}
\end{figure*}

One feature of the procedure we test is that it checks, and rejects, 99 models in favor of the model that has lowest disparity on an evaluation dataset. We find this method, even when it consistently reduces disparity out-of-sample, can increase the likelihood of \textit{passing over} the model that ultimately achieves lowest disparities out-of-sample. The frequency that the selected model is disparity-optimal of those considered is displayed in the rightmost plot in Figure \ref{fig:empirics-rf-sample-adult}. The results suggest that reasonable and appropriate search procedures sometimes require considering and rejecting models that end up being less discriminatory on unseen data. These results imply that it may not be a good idea to punish firms for having \textit{considered and rejected} models that are ultimately less discriminatory, as has been suggested by some scholarship \cite{scherer2019applying}. Even if a firm wishes to search for LDAs in good faith, doing so can often mean considering and rejecting models that end up having lower disparity.

The methods tested are not the only way to systematically search for less discriminatory algorithms. We could imagine altering the loss function of an algorithm directly to optimize some weighted combination of disparity reduction and utility. We could similarly imagine encoding knowledge about the structure of disparity. If certain attributes are likely proxies for the protected group belonging, using context-specific modeling assumptions could yield better de-biasing interventions.
\section{Conclusion}

This paper is concerned with a legal concept that predates the era of pervasive machine learning. The less discriminatory alternative has emerged as an enticing notion because it could offer a way to use existing legal concepts to reduce discrimination, and potentially at a negligible cost to firms. 
By attempting to more formally define less discriminatory \textit{algorithms}, we highlight how statistical, mathematical, and computational nuances inform what we should expect to be able to achieve with them. 
Our work suggests a number of natural theoretical, empirical, and legal directions for future work. 
On the theoretical side, future work could study model comparisons in the absence of held-out data, particularly when models are drawn from different distributions. 
Empirical studies of the benefits of LDA searches in real-world employment and lending settings would provide insights into their relative costs and benefits.
And finally, legal scholarship could provide clarity on how to develop standards for ``reasonable'' LDAs. 
Ultimately, our work seeks to bridge the gap between legal and algorithmic perspectives on less discriminatory alternatives by developing formal definitions that remain sensitive to the nuances in both fields.

\section{Acknowledgments}
The authors would like to thank Emily Black, Mingwei Hsu, Pauline Kim, Jon Kleinberg, Logan Koepke, Helen Nissenbaum, Lesia Semenova, and the members of the Fairness, Accountability, Transparency, and Ethics group (FATE) at Microsoft Research, the AI, Policy and Practice group (AIPP) at Cornell University, and the Digital Life Initiative (DLI) at Cornell Tech for providing feedback on this work. The work is supported in part by a grant from the John D. and Catherine T. MacArthur Foundation. Much of the work was completed while Laufer was an intern at Miscrosoft Research. He is also supported by a LinkedIn-Bowers CIS PhD Fellowship, a DLI doctoral fellowship, and a SaTC NSF grant CNS-1704527. Any opinions, findings, conclusions, or recommendations expressed in this material are those of the authors and do not reflect the views of NSF or other funding agencies.

\bibliographystyle{plainnat}
\bibliography{bibliography}
\appendix
\section{Further Related Work}

The so-called fairness-accuracy trade-off has been discussed by a number of scholars, and there are instances where improving both is attainable algorithmically \cite{rodolfa2021empirical, wick2019unlocking, chen2018my, petersen2023path}. The extent to which some form of disparate treatment is required to achieve demographic parity on datasets with base rate differences is discussed by \cite{lohaus2022two, castelnovo2022clarification}. Reasoning about what notions of fairness (e.g. selection rates) and utility are attainable relate to, broadly, the opportunities and limits of phrasing the LDA concept as an optimization problem, and face constraints about what is easily measurable and accessible (see e.g., \cite{laufer2023optimization, raghavan2024limitations}). \citet{fuster2022predictably} have explored the distributional effects \textit{between-group} of introducing machine learning to credit markets. 
\section{Supplementary Materials on Mathematical Limits}
\label{app:polygon}

Here we provide supplementary matrials and proofs for claims made in Section \ref{sec:polygon}.

\subsection{Proof of Theorem \ref{claim:boundonfairness}}

\begin{proof}

The proof relies on reasoning about the Pareto-efficient region of achievable \textbf{U} and $\mathbf{\Delta}$ values. We say an alternative classifier \textbf{Pareto-dominates} a starting classifier if 1) the alternative's performance measures (e.g., utility and disparity) are greater than or equal to the starting classifier's, and 2) at least one of the alternative's performance measures is strictly greater than the starting classifier's. We define the set of \textbf{Pareto-efficient} classifiers as the set of classifiers that are not Pareto-dominated. 

A classifier in our setting is specified by four quantities: the proportion of the population of type $(g,y)$ selected, for $g\in \{1,2\}$ and $y \in \{+,-\}$. Phrased this way, we can use the fact that we're searching over a convex decision space (a cube representing four proportion values) and any perturbation or `swap' affects disparity and utility as a function of the $g$ and $y$ value of the swapped data, as well as the fixed population sizes $n_{g,y}$.

Observe that the perfect classifier, $h^*(x)=y$, is in the Pareto-efficient set because it uniquely maximizes utility. We can use this classifier as an `anchor point' and find the alternative classifiers that optimally trade off utility and disparity. Notice there are exactly four possible ways to perturb $h^*$: a) Decrease the proportion of $(1,+)$ members with positive label, b) decrease the proportion of $(2,+)$ members with positive label, c) increase the proportion of $(1,-)$ members with positive label, or d) increase the proportion of $(2,-)$ members with positive label. We can immediately rule out two of these perturbations by noticing they strictly harm both accuracy and disparity: (b) and (c) introduce errors that strictly leave group 2 with fewer selections or group 1 with more selections (respectively). This leaves two possible candidates for introducing errors that optimally trade-off accuracy and disparity.

Two claims are necessary for the remainder of the proof. \textbf{First}, observe that \textit{there are sufficient swaps of each type (a) and (d) to achieve zero-disparity using only one of these types}.  This observations follows directly from our definitions of demographic disparity and `advantaged group.' If group 1 is advantaged (given), then $\texttt{SR}_1(h^*)>\texttt{SR}_2(h^*) \rightarrow \frac{\sum_x h^*(x|g=1)}{n_1} > \frac{\sum_x h^*(x|g=2)}{n_2}$. Adding given constraints, we can say $1 \geq \frac{\sum_x h^*(x|g=1)}{n_1} > \frac{\sum_x h^*(x|g=2)}{n_2} \geq 0$. However, based on the feasible set of classifiers, we know that by exhausting all of swap (a) we'd have $\texttt{SR}_1=0$ without any change to $\texttt{SR}_2$. Similarly, if we exhaust all swaps of type (d) we'd have $\texttt{SR}_2=1$ at no change to $\texttt{SR}_1$. Because there are sufficient swaps to change the ordering of the group selection rates, we know that there are enough of each swap type to equalize selection rates across groups. 

\textbf{Second}, observe that with full information, each swap type has a constant ratio between utility reduction and disparity reduction, given as follows:
\begin{itemize}
    \item Swap type (a): Reducing the number of positively-labeled members of group 1 corresponds to a decrease in the true positive rate of $\frac{1}{n_+}$ meaning the unit change in $\mathbf{U}$ as a function of swap (a) is $-\frac{1}{n_+}$. The effect on disparity is $\frac{1}{n_1}$. So the ratio between utility reduction and disparity reduction is $\frac{n_1}{n_+}$.
    \item Swap type (d): Increasing the number of positively-labeled members of group 2 corresponds to an increase in the false positive rate of $\frac{1}{n_-}$ meaning the unit change in $\mathbf{U}$ as a function of swap (d) is $-\lambda\frac{1}{n_+}$. The effect on disparity is $\frac{1}{n_2}$. So the ratio between utility reduction and disparity reduction is $\frac{\lambda n_2}{n_-}$.
\end{itemize}

Taken together, the above two claims suggest that a constant (real, not necessarily integer) number of `swaps' draw the Pareto-efficient region trading off between utility and disparity, and this region is a line segment defined by two points in utility-disparity space: utility-optimal classifier $h^*$ and the utility-optimal classifier with zero disparity $h^f$, representing by $(\mathbf{\Delta,U})=(0,u^f)$, where $u^f= 1-\min\left[\frac{n_1}{n_+},\frac{\lambda n_2}{n_-}\right]\mathbf{\Delta}(h^*)$. 

For any starting classifier with utility less than or equal to $u^f$, the classifier $h^f$ represents a zero-disparity alternative meeting the utility requirements. Otherwise, the minimum-disparity classifier at the a given utility value corresponds to the point in the Pareto region at the same starting utility value.

\end{proof}
\section{Supplementary Materials on Computability}

Here we provide supplementary materials and proofs for claims made in Section \ref{sec:computation}. 

\subsection{Intuition for Hardness Result}

Here, we provide some intuition for the proof of the computational complexity of the full information LDA. Consider the specific case where our dataset is made up of discrete and finite data values. The set of possible values is 
 $\mathbf{X} = \{x_1,x_2,...,x_{|\mathbf{X}|}\}$. At its core, the mathematical task of finding the LDA is to `grab' (i.e., positively label) a subset of data values in $\mathbf{X}$ that meet certain utility and disparity requirements. Intuitively, we can imagine our data set as a collection of points. Every point $x_i \in \mathbf{X}$ has a disparity value $d_i$ and a utility value $u_i$. We are looking to find a subset $S^* \subseteq \mathbf{X}$ such that the classifier $h'(x)=\mathbf{1}[x\in S^*]$ minimizes disparity $\mathbf{\Delta}(h') = \left|\sum_{i \in S} d_i\right|$ subject to a single constraint $\mathbf{U}(h') = \sum_{i \in S^*} u_i \geq \mathbf{U}(h^0)$, where we can think of $u(h^0)$ as some constant utility cutoff or threshold that constrains the search. Notice that if we switch the sign on utility, we can think of each point as having a \textit{cost} $c_i:=-u_i$, and our utility constraint can be thought of as a budget $B:=-\mathbf{U}(h^0)$ that constrains how much we are able to spend.

We're left with the following optimization problem formulation for finding the LDA:
$$\min_{S \subseteq \mathbf{X}} \left|\sum_{i} d_i \mathbf{1}[x_i \in S] \right| \ \ \text{s.t. } \sum_{i=0}^N c_i \mathbf{1}[x_i \in S] \leq B. 
$$

If there exists a solution to the above optimization problem, then it would qualify as an LDA and solve the problem put forward in Definition \ref{def:full-info-LDA}. However, solving this problem may be computationally onerous. Notice the optimization above looks almost identical to the 0-1 Knapsack Problem, which is known to be NP-hard. A noticeable difference in our case, however, is that we take an absolute value over the objective function. The full proof (in Appendix \ref{app:NP-hard-proof}) involves a reduction from a related problem, the Subset Sum Problem.

\subsection{Proof of Theorem \ref{thm:NP-complete}}
\label{app:NP-hard-proof}

\begin{proof}

    We provide a polynomial-time reduction from the subset sum problem (known to be NP-complete) to the LDA problem. We use a particular case of the subset-sum problem (preserving the NP-hard property), defined below:
    
    \begin{definition}[Subset sum problem]
        Given a set $\mathcal{W}$ of integers $\{w_1,w_2,...,w_{|\mathcal{W}|}\}$, find a subset whose values sum to $0$.
    \end{definition}

    Suppose we have a black box that computes any LDA problem. We're given an instance $\langle \mathcal{W} \rangle$ of the subset sum problem. Our goal is to build an instance of the LDA problem where the solution enables us to solve the subset sum problem efficiently.

    \paragraph*{Building a population.} To start the proof, we will build a population of people that collectively compose our distribution. We can construct this population however we want, and at the end, we'll have an LDA search over this population where the LDA classifier (if one exists) will precisely identify the subset of integers that solve subset sum. The people will be split cleanly into categories, denoted by a categorical data value $x \in \mathbf{X}$, which takes values $\{1,2,...\}$. We construct the population as follows: cycle through the values in our subset sum problem $\{w_1,...,w_{|\mathcal{W}|}\}$. For each element $i$, if the integer $w_i$ is positive, add $2w_i$ people of group 1 to our population, each of type $x=i$. Otherwise, if integer $w_i$ is negative, add $2|w_i|$ people  of group 2 to our population, each of type $x=i$. We end up with a population that looks something like Table \ref{tab:preliminary-pop}.

    \begin{table}[h!]
        \centering
        \begin{tabular}{c|c|c}
    \hline
    \hline
        $\mathbf{X}$ & $G$ & $n_g(x)$ \\
        \hline
        1 & 1 & $2w_1$ \\
        1 & 2 & 0\\
        2 & 1 & 0  \\
        2 & 2 & $2|w_2|$  \\
        \vdots & \vdots & \vdots   \\
        $|\mathcal{W}|$ & 1 & $2w_{|\mathcal{W}|}$  \\
        $|\mathcal{W}|$ & 2 & 0  \\
        \hline
        \hline
    \end{tabular}
        \caption{Preliminary example of a population}
        \label{tab:preliminary-pop}
    \end{table}
    
    In the example in Table \ref{tab:preliminary-pop}, the first subset sum value $w_1$ is positive, the second is negative, and the last is positive, which means that the people are assigned to groups 1, 2, and 1, respectively. We use $n_g(x)$ to refer to the number of people of data type $x$ and group $g$.

    To complete our population, we add two more data values, $x^*$ and $x^{**}$. The first data value, $x^*$, contains a \textit{single} person of group 1. The second data value, $x^{**}$, will equalize the total number of people of each group: we take the absolute difference between all the people we've created in group 1 and all the people we've created in group 2, which is equal to $\left|2\sum_i w_i + 1\right| $. We'll add this many people with data type $x^{**}$ to whichever group has fewer people to equalize the overall number of people in each group. Finally, if we want, we can add additional people to the population, but they must have data type $x^{**}$ and an equal number must be added from each group, to preserve balance. So we'll say formally that $2\alpha$ people may be added for some non-negative integer $\alpha$ as long as exactly $\alpha$ are added from each of groups 1 and 2. The total number of people in our population is therefore $N:= 2\sum_i |w_i| + 1 +  \left|2\sum_i w_i + 1\right| + 2\alpha$.

    \begin{table*}[ht]
        \centering
        \begin{tabular}{c|c|c|c|c|c|c}
        \hline
        \hline
        $\mathbf{X}$ & $G$ & $n_g(x)$ & $\rho_g(x)
        $ & $\sigma(x)
        $  & $U(x)$ & $h^0(x)$\\
        \hline
        \rule{0pt}{4ex}    
        1 & \makecell{1\\2} & \makecell{$2w_1$\\0} & \makecell{$\frac{4}{N}w_1$\\0} & 1 & $\frac{2}{N}w_1$ & 0\\
        2 & \makecell{1\\2} & \makecell{0\\$2|w_2|$} & \makecell{0\\$\frac{4}{N}|w_2|$} & 1 & $\frac{2}{N}|w_2|$ & 0\\
        \vdots & \vdots & \vdots & \vdots & \vdots & \vdots & \vdots   \\
        $|\mathcal{W}|$ & \makecell{1\\2} & \makecell{$2w_{|\mathcal{W}|}$\\0} & \makecell{$\frac{4}{N}w_{|\mathcal{W}|}$\\0} & 1 & $\frac{2}{N}w_{|\mathcal{W}|}$ & 0 \\
        \hline
        \rule{0pt}{4ex}    
        $x^*$ & \makecell{1\\2} & \makecell{1\\0} & \makecell{$\frac{2}{N}$\\0} & 1 & $\frac{1}{N}$  & 1 \\
        \hline
        \rule{0pt}{4ex}
        $x^{**}$ & \makecell{1\\2} & \makecell{$\alpha$\\\tiny{$\left|2\sum_i w_i + 1\right| + \alpha$}}  & \makecell{$\frac{2\alpha}{N}$\\\tiny{$\frac{2}{N}\left(\left|2\sum_i w_i + 1\right| + \alpha\right)$}} & 0 & \tiny{$-\lambda \frac{1}{N}\left(\left|2\sum_i w_i + 1\right| + 2\alpha\right)$} & 0\\
        \hline
        \hline
    \end{tabular}
        \caption{Table representing the constructed population.}
        \label{tab:hardness-result-population}
    \end{table*}

    Now, our population is fully specified. Since we know the total number of people ($N$), we are now able to specify the frequency (or probability density) of each data value in our population, both overall and within group --- the distribution is specified by $\rho_g(x)$. Our last step is to specify the outcome labels (or base rates) 
for each data value in our population. For this, we simply assign $\sigma(x)=1$ for all $x\neq x^{**}$, and $\sigma(x^{**})=0$.

    Our completed population is depicted in Table \ref{tab:hardness-result-population}.

    Based on the way we've set up this population, we want the LDA search to select exactly the values in $\mathbf{X}$ that correspond to the indices of $\mathcal{W}$ that solve the subset sum instance (if and only if a solution exists). Therefore, we want the LDA search to never select our slack parameters $x^*, x^{**}$. The first slack parameter will never be selected because the density of the group-1 population in $x^*$ is $\frac{2}{N}$, which is half the minimum difference between any two densities of group 2, so a LDA with disparity 0 could never include data $x^*$. The second slack parameter will never be selected as long as the utility of selecting it is negative with greater magnitude than the sum of all other utilities in the entire search, since in that case, including it would always yield a utility less than $0<U(h^0)=\frac{1}{N}$. Thus we require: $U(x^{**})< -\left|\sum_{x \neq x^{**}}U(x)\right|$. We do some arithmetic on this condition and simplify to a requirement on the value $\alpha$: 
    \begin{eqnarray*}
        U(x^{**})< -\left|\sum_{x \neq x^{**}}U(x)\right| \\
        -\lambda \frac{1}{N}\left(\left|2\sum_i w_i + 1\right| + 2\alpha\right) < -\frac{2}{N} \sum_i |w_i| - \frac{1}{N}\\
        \lambda \left|2\sum_i w_i + 1\right| + 2\lambda \alpha > 2 \sum_i |w_i| + 1 \\
        2\lambda \alpha > 2 \sum_i |w_i| + 1 - \lambda \left|2\sum_i w_i + 1\right| \\
        \alpha > \frac{1}{\lambda}\sum_i |w_i| + \frac{1}{2\lambda} - \left|\sum_i w_i + \frac{1}{2}\right|  \\
        \alpha > \frac{1}{\lambda} \left(\sum_i |w_i| + \frac{1}{2}\right)
    \end{eqnarray*}
    The above inequality implies that our hardness result does not depend on any particular value of $\lambda$. We can complete the reduction for any feasible value $\lambda>0$, as long as we set $\alpha$ to be sufficiently large. Now we're in a position to state the following Lemma, which represents the rest of what is needed for our proof:

    \begin{lemma}
        Consider the LDA problem $\langle \mathbf{X},\rho_g, \sigma, h^0 \rangle$, for any given value $\lambda>0$. If an LDA does not exist, there is no solution to subset sum problem $\langle \mathcal{W}\rangle$. If an LDA $h^*$ does exist, the indices of the data for which $h^*(x_i)=1$ are the subset $\{w_i\} \subseteq \mathcal{W}$ which sums to 0.
    \end{lemma}

    We prove this Lemma with the following sequence of claims:

    \begin{claim}
        The LDA classifier will never select $x^{**}$.
    \end{claim}
    \begin{proof}
        The utility from labeling $x^{**}$ positively is a negative value that is strictly less than the sum of all other utilities. The utility of our starting classifier is $\frac{1}{N}>0$. So, no classifier can positively label $x^{**}$ and have at least as good utility as the starting classifier.
    \end{proof}
    
    \begin{claim}
        The LDA classifier will never select $x^*$.
    \end{claim}
    \begin{proof}
        We've already established the LDA will never positively label $x^{**}$. If there exists an LDA, the disparity must be (strictly) less than the original classifier. But notice that every other data point has disparity value divisible by $\frac{4}{N}$. As the sole data point with disparity value $\frac{2}{N}$, including $x^*$ must inevitably yield disparity at least $\frac{2}{N}$ rendering an LDA impossible.
    \end{proof}

    \begin{claim}
        An LDA exists if and only if the disparity values of positively-labeled data sum to 0.
    \end{claim}
    \begin{proof}
        The total disparity of any classifier can by calculated by summing the densities $\rho_1(x)-\rho_2(x)$ of all positively-classified data and then taking an absolute value. Notice that aside from our slack variables $\{x^*,x^{**}\}$, all other data points' disparities are divisible by $\frac{4}{N}$. Thus, any sum of disparities from this group of data must also be divisible by $\frac{4}{N}$, meaning the least discriminatory grouping must have one of the following disparities: $0, \frac{4}{N}, \frac{8}{N}, \frac{12}{N},...$ Now, notice the disparity of our starting classifier $h^0$ is $\frac{1}{N}$. This value is strictly between 0 and the minimum non-zero disparity given our setup, $\frac{4}{N}$. Thus, any less-discriminatory classifier that is attainable from the non-slack variables ($x \in \{1,2,...,|\mathcal{W}|\}$) must have disparity values that sum to 0. Finally, we've already shown that no LDA solution will include a positive label for the slack variables $x^*, x^{**}$, so this concludes the proof.
    \end{proof}

    \begin{claim}
        The subset sum problem is invariant to a scaling factor $k$.
    \end{claim}
    \begin{proof}
         This is a property of any summation. $a+b=c \iff ka +kb = kc$. In our case, the target is $0$ so $k*0=0$. The scaling factor we use is $k=\frac{2}{N}$ but this is easy to verify for any real-valued $k$.
    \end{proof}

    This concludes the reduction.

    We've thus shown the full-information LDA problem is at least as hard as Subset Sum. Showing that the full-information LDA problem is NP-complete further requires that the problem is in NP. We know this is true because, given a candidate solution and a problem statement, we can simply check whether the solution has lower demographic disparity and greater or equal utility (which both require only taking a sum over the population).
\end{proof}

\subsection{A $(1+\epsilon)$ Approximation}
\label{app:LDA-approximation}

We have shown that the task of finding a less discriminatory alternative classifier is NP-hard, even in cases where we can access the true probability of a positive outcome for every data point. However, hardness does not imply impossibility, and there remain strategies for identifying LDAs. Here, we ask, what if we're okay with identifying a less discriminatory alternative so long as the starting classifier $h^0$ is at least some minimum distance from the \textit{least} discriminatory alternative classifier?

Hypothetically, consider a case where the difference in disparity between a classifier $h^0$ and the least discriminatory alternative (which we can call $h''$) is 0.01. We may decide that such small differences are legally and ethically \textit{de minimis}, i.e., too trivial or small to merit consideration. Equipped with some minimal value of this sort, denoted $\epsilon$, perhaps we'd be satisfied if we can reliably and efficiently identify an alternative classifier whose disparity is within a factor of at least $(1+\epsilon)$ of the least discriminatory possible alternative. If we have this approximation, then the \textit{only case} where we might fail to find an LDA where one exists is when the starting classifier $h^0$ is exceedingly close to the least discriminatory alternative $h''$. In cases where the starting classifier is outside of a factor of $(1+\epsilon)$ of $h''$, we can guarantee that we will find a less discriminatory alternative.

\begin{definition}[Full-information approximate LDA]
    \label{def:approx-LDA}
    Given a set of data values $\mathbf{X}$, a Bayes-optimal predictor $\sigma(x)$, a group-specific probability density function $\rho_g(x)$, a utility function $\mathbf{U}(h;\lambda)$, a baseline classifier $h^0(x)$, 
    and $\epsilon >0$, a \textbf{full-information $\epsilon$-approximate LDA} is a classifier $h'$ with utility at least $\mathbf{U}(h^0)$ and absolute disparity less than $|\mathbf{\Delta}(h^0)|(\frac{1}{1+\epsilon})$.
\end{definition}

\begin{claim}
    The full-information $\epsilon$-approximate LDA can be identified in polynomial running time bounded by $O(n^3\epsilon^{-1})$.
\end{claim}

This claim's proof can be thought of as an exercise --- and the proof is deferred --- given its resemblance to the subset sum algorithm and approximation scheme \cite{kleinberg2006algorithm}. We implement the algorithm and its empirical performance is discussed in the next section.

\subsection{Empirical Demonstration of Approximation Algorithm}
\label{app:LDA-approximation-empirics}

\begin{figure*}[t]
    \centering
    \includegraphics[width=1\linewidth]{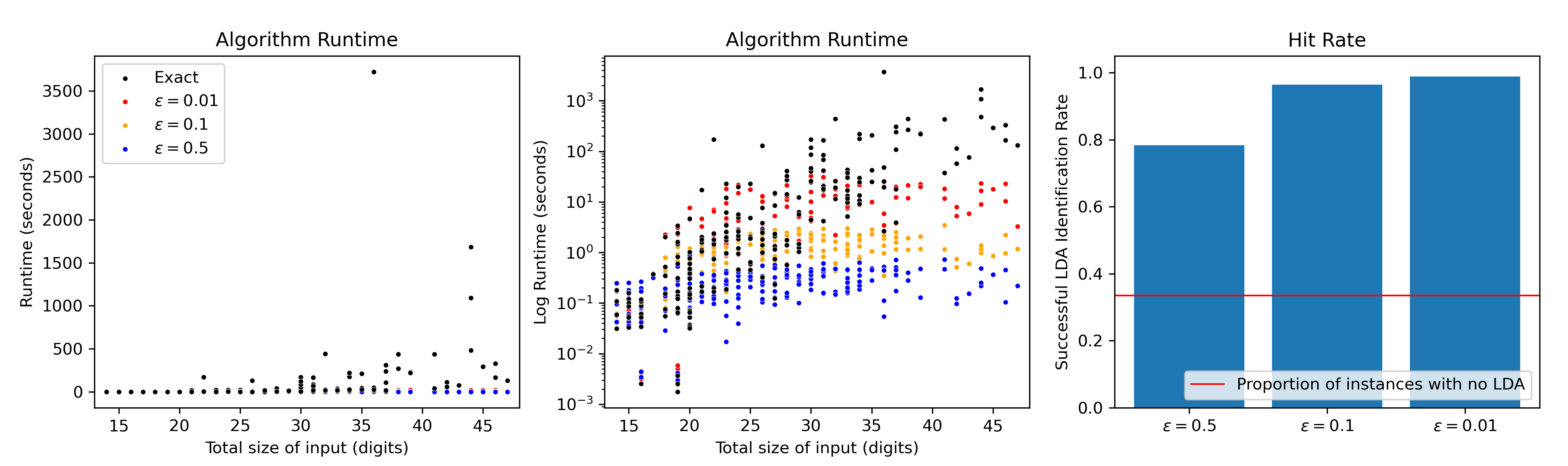}
    \caption{Runtime (left, center) and accuracy performance (right) of exact and approximate algorithms for computing the least discriminatory alternative classifier, given oracle access to information about the population distribution. The exact algorithm is compared to three different $(1+\epsilon)$ approximation algorithms with varying $\epsilon$. 
    A total of $250$ instances of the LDA problem were randomly generated, with varying numbers of data types, starting classifiers, and numbers of digits used to specify probability densities. For every generated instance, all four algorithms were run and runtime and accuracy were recorded. As the size of the input increases (measured as the total number of digits beyond the decimal places used to specify $\sigma(x)$ and $\rho_g(x)$), the worst-case runtime complexity for the exact algorithm explodes. However, polynomial-time approximations achieve a high hit rate, defined as the rate at which an LDA is identified successfully.}
    \label{fig:FPTAS}
\end{figure*}

To demonstrate the effectiveness of the approximation algorithm empirically, we implement both exact and approximation algorithms to solve any LDA problem, specified by $\rho_g(x), \sigma(x), h^0(x)$, for some discrete number of values $x\in \mathbf{X}$. Equipped with these algorithms, we can evaluate their performance on instances of the LDA problem. These findings corroborate the formal results put forward in this section: that although the worst-case runtime explodes for a general instances of the LDA, there are efficient (i.e., polynomial-time) approximation algorithms that identify LDAs with high accuracy.

\paragraph*{Simulating LDA problems.} To simulate instances of the LDA problem, we first specify the number of discrete data values, $$\texttt{n}_{\texttt{options}} := |\mathbf{X}|,$$ and the maximum number of digits per density specification, \texttt{max\-digits}. Given these specifications, a simulator generates two sets of uniform random numbers, each of size $\texttt{n}_{\texttt{options}}$, and then performs a manipulation so that each set adds to 1 and contains float values with at most \texttt{max\-digits} digits after the decimal place. These sets are used as values $\rho_1$ and $\rho_2$, the group-specific probability densities over the data. Finally, using the identical procedure, a simulator establishes the true probability values for each data option, $\sigma(x)$, this time using at most 2 digits after the decimal place. Finally, to create variation in the total size of the input, we simulate instances with varying values for $\texttt{n}_{\texttt{options}} \in \{4,5\}$ and $\texttt{max\-digits}\in\{1,2,3,4\}$. In total, we generate $202$ instances of LDA problems, of which 164 contain an identifiable less discriminatory alternative and 86 contain no LDA. In this simulation, the input sizes range from 14 to 47 total decimal digits.

The runtime and accuracy performances are displayed in Figure \ref{fig:FPTAS}. The code is additionally available in an \href{https://github.com/bendlaufer/what-constitutes-LDA}{online repository}.

\section{Supplementary Materials on Empirics}
\label{app:empirics}

\textbf{Data and Models (further information)}. Every row in the \texttt{Adult} dataset represents features of an individual, and the outcome variable is an indicator representing whether they make over \$50,000 in annual income. The data is on American adults from census information in 1994. We use a total of four variables -- \texttt{maritalstatus} (a categorical variable representing whether the individual is single/divorced/married/etc), \texttt{hoursperweek} (a numerical variable representing the number of hours the individual works per week), \texttt{education} (a categorical variable specifying the level of education achieved, e.g. college degree), and \texttt{workclass} (a categorical variable representing the type of work). The protected attribute is the \texttt{gender} (given as a binary M/F). The dataset contains a total of 32,561 rows.

The \texttt{German Credit} dataset contains loan decisions in Germany, and was accessed via the UCI machine learning database. The dataset contains a total of 1,000 rows with categorical and numerical features related to individual's credit, financial status, employment, and loan application. The protected attribute is the gender (given as a binary M/F). The following five features were used to train the models: \texttt{credit\_history\_category} (categorical variable representing information about credit history), \texttt{credit\_amount} (numerical variable representing amount of credit requested by loan applicant),  \texttt{unemployment\_category} (categorical variable with information about whether the individual is unemployed), \texttt{install\-ment\_rate\_per\-centage\_income} (numerical feature representing the ratio between loan amount and income), and \texttt{pre\-sent\_res\-id\-ence\_dur\-ation} (a numerical variable representing the amount of time the applicant has resided in their current residence).

\end{document}